\documentclass[twoside,leqno,twocolumn,10pt]{article}
\usepackage[letterpaper]{geometry}

\usepackage{ltexpprt}

\usepackage{microtype}

\usepackage{xcolor}
\usepackage{multirow}
\usepackage{url}
\usepackage{algorithm}
\usepackage[noend]{algpseudocode}
\usepackage{algpseudocode}
\usepackage{subfigure}
\usepackage{enumitem} 
\usepackage{times}
\usepackage{xspace}
\usepackage{latexsym}
\usepackage{amsmath}
\usepackage{amsfonts}

\usepackage{tabularx}
\usepackage{balance}
\usepackage{subfigure}
\usepackage{siunitx}

\usepackage{booktabs}
\usepackage{colortbl}
\usepackage{float}
\usepackage{listings}
\usepackage[bookmarks=false,draft]{hyperref}

\usepackage{mathtools}
\usepackage{pgfplots}
\usepackage{MnSymbol}

\setlist{nolistsep}

\DeclareSIUnit[scientific-notation=engineering,prefixes-as-symbols=false]{\ns}{\nano\second}

\newcommand{\para}[1]{\smallskip\noindent\textbf{#1}}

\newcommand{\cut}[1]{}

\newcommand\code[1]{\texttt{\lstinline[mathescape]$#1$}}

\newcounter{algsubstate}
\renewcommand{\thealgsubstate}{\alph{algsubstate}}

\newcounter{insightlabel}
\newcounter{insightnmbr}
\renewcommand{\theinsightlabel}{\textbf{\theinsightnmbr}}

\newcommand{\tightcaption}[1]{\vspace{-0.2cm}\caption{#1}\vspace{-0.2cm}}

\usepackage{comment}
\usepackage{dsfont}
\usepackage{bm}
\usepackage[noend]{algpseudocode}
\usepackage{float}
\usepackage{multirow}
\usepackage{hhline}
\usepackage{calc}
\usepackage{graphicx}
\usepackage{framed}
\usepackage{xspace}
\usepackage{algorithm}
\usepackage{array}

\newtheorem{thm}{Theorem}
\newtheorem{lem}[thm]{Lemma}
\newtheorem{definition}[thm]{Definition}

\newcommand{\var}{\mathsf{Var}}

\def\_{\,\,\,\,\,}

\def\RTT{\mathsf{RTT}}
\def\TPL{\mathsf{TPL}}

\newcommand{\eps}{\epsilon}

 \def\Pr{{\mathrm{Pr}}}

\newcommand{\EEx}[1]{\ensuremath{\mathbb{E}\left[#1\right]}}
\renewcommand{\O}[1]{\ensuremath{\mathcal{O}\left(#1\right)}}

\usepackage{multicol}

\DeclareMathOperator*{\argmin}{argmin}
\DeclareMathOperator*{\polylog}{polylog}

\newenvironment{packeditemize}{
\begin{list}{$\bullet$}{
\setlength{\itemsep}{1.5pt}
\setlength{\labelwidth}{8pt}
\setlength{\leftmargin}{10pt}
\setlength{\labelsep}{3pt}
\setlength{\listparindent}{\parindent}
\setlength{\parsep}{1.5pt}
\setlength{\parskip}{1.5pt}
\setlength{\topsep}{1.5pt}}}{\end{list}}
\begin{document}
\title{Memory-Efficient Performance Monitoring on Programmable Switches with Lean Algorithms}

\author{Zaoxing Liu$^1$, Samson Zhou$^1$, Ori Rottenstreich$^2$, Vladimir Braverman$^3$, Jennifer Rexford$^4$\\
\normalsize{$^1$Carnegie Mellon University, $^2$Technion, $^3$Johns Hopkins University, $^4$Princeton University}}

\date{}

\maketitle

\begin{abstract}
Network performance problems are notoriously difficult to diagnose. 
Prior profiling systems collect performance statistics by keeping information about each network flow, but maintaining per-flow state is not scalable on resource-constrained NIC and switch hardware.  
Instead, we propose sketch-based performance monitoring using memory that is sublinear in the number of flows. 
Existing sketches estimate flow monitoring metrics based on flow sizes. 
In contrast, performance monitoring typically requires combining information across pairs of packets, such as matching a data packet with its acknowledgment to compute a round-trip time. 
We define a new class of \emph{lean} algorithms that use memory sublinear in both the size of input data and the number of flows. 
We then introduce lean algorithms for a set of important statistics, such as identifying flows with high latency, loss, out-of-order, or retransmitted packets.
We implement prototypes of our lean algorithms on a commodity programmable switch using the P4 language. 
Our experiments show that lean algorithms detect $\sim$82\% of top 100 problematic flows among real-world packet traces using just 40KB memory.
\end{abstract}
\section{Introduction}
Modern datacenter operators require timely and accurate information about the performance of the underlying network to optimize their network services. 
High bandwidth usage and efficiency of the network infrastructure are essential for cloud provider's cost control and consumer satisfaction. 
However, due to temporary congestion, link failure, or adversarial traffic, the performance of the network can quickly degrade. 
Thus, cloud providers seek efficient ways to evaluate and maintain the quality of their network infrastructure. 

Traditional techniques for diagnosing performance problems rely on offline analysis of traffic traces~\cite{tcptrace,critical_tcp,Zhang02,Wu2013}. 
However, these offline solutions are not capable of real-time diagnosis and incur significant data-collection overhead. 
In this paper, we focus on \emph{online} performance analysis. 
When conducting diagnosis, an operator needs to track multiple statistics, which are important to detect performance degradation and pinpoint problematic issues. 
For instance, existing tools track TCP sending and receiving window sizes~\cite{snap,marple,dapper}, lost packets~\cite{marple,dapper}, round-trip latency~\cite{marple}, out-of-order packets~\cite{dapper}, and re-transmitted packets in order to understand the current network quality. 

Prior efforts on online performance analysis usually require end-host access.  For instance, SNAP~\cite{snap} and HONE~\cite{hone,perfCDN} modify end-host network stack to collect TCP statistics. 
These solutions offer good monitoring capabilities when full control of all the network components is granted. 
However, modifying the end-host's operating system is not ideal on the following two fronts: 
(1) In public clouds, operators cannot alter the end-host's network stack without undermining the isolation of tenant's virtual machines (VMs). 
(2) In other types of networks, operators do not even have access to end-user machines, e.g., in a carrier network. In both settings, measuring performance statistics \emph{inside} the network is a better approach. 

\begin{figure}[t]
\centering
\includegraphics[width=0.98\linewidth]{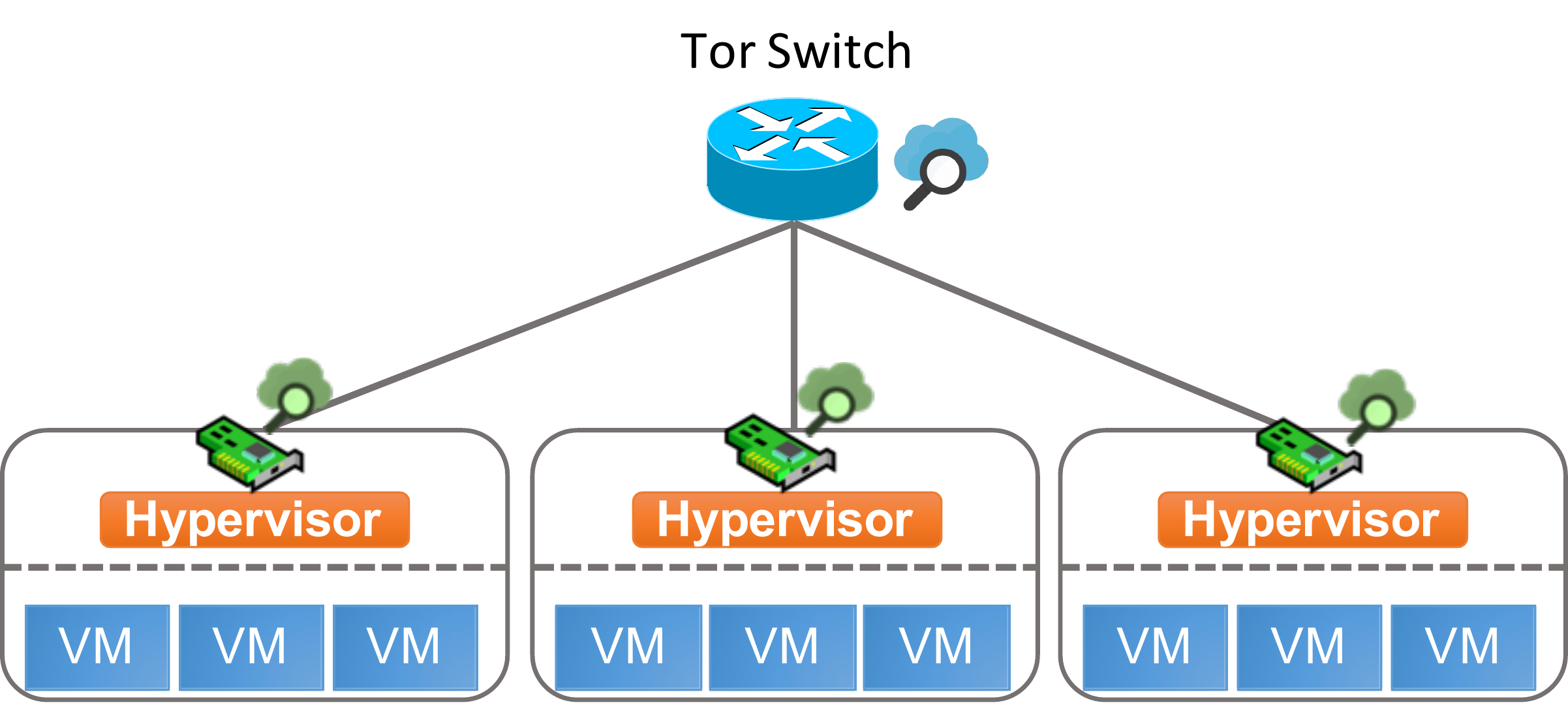}
\tightcaption{Performance monitoring deployed at the edge.}\label{fig:edge}
\vspace{-3mm}
\end{figure}

In this context, Marple~\cite{marple} is a language to query performance statistics from programmable switches and Dapper~\cite{dapper} is a recent tool for TCP performance diagnosis deployed at the ``edge''---hypervisor, NIC, or top-of-rack switch, adjacent to the end-host.
Unfortunately, network devices (e.g., NICs and switches, as shown in Figure~\ref{fig:edge}) have limited memory and computation power available for handling performance measurement tasks. 
Tools like Marple~\cite{marple} and Dapper~\cite{dapper} rely on per-flow data structures such as per-flow counters and recorded per-packet timestamps. 
These data-plane monitoring techniques become infeasible when the number of connections is large, particularly as increasing link speeds lead to ever more flows on each link.

To this end, we ask \emph{can we build a performance monitoring tool with high accuracy that is memory-efficient}? 
\begin{definition}
We define the performance-monitoring algorithms that use memory sublinear in both the size of input data and the number of flows\footnote{We follow common terminology, referring to a ``flow'' as the packets sharing a 5-tuple (SrcIP, DstIP, SrcPort, DstPort, Proto).
In practice, a flow can also be defined as an Origin-Destination pair or other combinations of packet header fields.} as lean algorithms. 
\end{definition}
In this work, we aim to design lean algorithms as a response to the question above. A natural observation is that we cannot achieve lean algorithms by simply tracking and maintaining accurate information for every flow. 
Instead, we attempt to measure the top-$k$ flows that contribute the most latency, packet loss, out-of-order packets, and retransmitted packets (i.e., the most ``influential'' flows). 
Top-$k$ influential flows are useful because, in the presence of network performance problems between end-hosts, these flows contribute significantly to the performance metric and are strong indicators reflecting the performance issues. For instance,  flows that have a number of packet losses are useful information for faulty link detection~\cite{netbouncer}.



Our first attempt to achieve lean algorithms is characterizing the 
performance-monitoring problems based on a redefined streaming model. 
In this model, we prove that if a performance-monitoring function meets the so-called \emph{flow-additive} property (in \S\ref{sec:model}) and can be approximated on an individual flow using space sublinear in the input data \emph{for that particular flow}, then we can provide a lean algorithm to detect the most influential flows defined by this function, using space sublinear in not only the total input size but also the number of flows in the network.  To the best of our knowledge, our work is the first to propose algorithms with such memory requirements.

Our design is inspired by the rich literature of \emph{sketching algorithms} in network flow monitoring, where the traffic is modeled as a stream of elements~\cite{ams}. 
A number of sketching algorithms have been introduced to accurately estimate various flow metrics such as heavy hitters~\cite{CMSketch,countsketch,SpaceSavings,HashPipe,HarrisonHH,OpenSketch,univmon}, hierarchical heavy hitters~\cite{MitzenmacherHHH,cormodeHHH,RHHH}, flow size distribution~\cite{CMSketch,Kumar:2004,FlowRadar}, and change detection~\cite{OpenSketch,univmon}. 
These algorithms allow for memory-efficient measurement systems while maintaining guaranteed fidelity. 
In a similar spirit, we leverage sketching techniques as a building block to achieve excellent memory efficiency for a new set of performance monitoring metrics. 
Unfortunately, we provably cannot provide a general solution for performance monitoring due to the different definitions between various performance metrics and the ways to measure them. For instance, latency is measured by the time of a packet sent and the time its ACK received, while packet loss is measured by the number of missing packets in the flow.
Specifically, we also demonstrate lower bounds that show lean algorithms for specific performance-monitoring problems cannot be determined without additional assumptions and problem relaxations. 

We empirically evaluate the accuracy and efficiency of our algorithms to detect the top 100 flows that contribute the largest portion of latency, packet loss, out-of-order-packets, and retransmitted packets. We implement a prototype in P4~\cite{P4} and evaluate on a 6.5Tbps Barefoot Tofino switch~\cite{tofino}. Our experiment results are based on analyzing a range of network traces with 2.3M to 3.7M flows and injected performance issues, and demonstrate good accuracy using small memory: $\sim$82\% accuracy with 40KB and $>$90\% accuracy with 160KB for all tested performance statistics. 

\para{Roadmap and contributions:} This paper is organized as follows:
\begin{itemize}
    \item We describe a set of typical performance statistics and their use cases in diagnosing performance issues. (\S\ref{sec:stat_bg})
    \item We show a viable path towards lean algorithms. We introduce a streaming computational model and show that lean algorithms can be achieved if the performance statistics satisfy the \emph{flow-additive} property. (\S\ref{sec:model})
    \item We propose and analyze four lean algorithms for tracking the flows that contribute most latency, packet loss, out-of-order packets, and retransmitted packets. (\S\ref{sec:algo})
    \item We implement a proof-of-concept prototype using P4 and optimize it for Barefoot Tofino hardware target. (\S\ref{sec:impl})
    \item Our trace-driven evaluation shows that our approach effectively detects the most influential flows among all flows with a tiny amount of memory. (\S\ref{sec:eval})
\end{itemize}

Finally, we discuss related work in \S\ref{sec:related} and highlight some future directions this work opens up in \S\ref{sec:conclude}.
\section{Network Performance Statistics}\label{sec:stat_bg}
Performance degradation in network connections can occur for many reasons. 
For instance, if a flow experiences high packet loss or latency, the congestion-control algorithm reduces the sending rate, leading to lower performance. 
To profile such performance issues, we need to measure performance statistics in a timely and efficient manner. 
In this section, we discuss representative statistics that are widely measured in diagnosing performance problems~\cite{snap,marple,dapper}. 
In the following statistics, we formally define a \emph{packet} as a tuple or a sub-tuple of $\{key, type, seq no., ack no., time\}$, where $key$ is the flow identity (e.g., 5 tuple), $type$ represents a specific packet type, $seq/ack no.$ is the sequence number or corresponding acknowledgment (ACK) number (we use relative numbers $1,2,3\dots$ for illustration in this paper), and $time$ is the timestamp.

\para{Round-trip latency:} 
The latency of a connection can often be measured by the difference between the packet transmission time and the receiving time of the corresponding response (i.e., round trip time (RTT)). 
Many network applications, such as online gaming or trading, demand fast responses to information about new events, and are therefore extremely sensitive to latency. 
Hence, minimizing network latency is expected of any adequate network management.

To this end, we identify the $k$ flows that contribute the highest round-trip times. 
The round-trip latency for a packet is measured as the time difference between a sent packet and its corresponding ACK; the total round-trip time of a flow is the sum of the latencies of the packets across the flow. 
Since the latency measurement relies on the timestamps of a pair of packets and a missing ACK may happen at any time in the network, we can specify a set of special packets to  measure the latency in order to minimize the probability of latency overestimation from missing ACKs.
For instance, a TCP SYN packet can be defined as \{key=(1.1.1.1, 2.2.2.2, 123, 80, TCP), type=SYN, seq=2, time=$t_1$\}, with its corresponding SYN-ACK packet defined as \{key=(2.2.2.2, 1.1.1.1, 80, 123, TCP), type=SYN-ACK, ack=3, time=$t_2$\}. 
Then the measured latency is $t_2-t_1$ units; when this measured latency increases, the actual throughput drops.

\textit{High network path latency to detect slow TCP rate:} The TCP sending rate depends on the window size and RTT, and the RTT measurement directly decides the TCP sending rate at the client-side. To diagnose the problem, when the RTT is higher than normal (e.g., an expected RTT from an operator or the minimum RTT among connections), we can decide if the latency on this connection is acceptable.

\para{Packet loss:} 
Flows with high packet loss can be used to identify the network routes with potential performance issues. 
For instance, if a recent window of packets received is \{key=(1.1.1.1, 2.2.2.2, 123, 80, TCP), seq=1\} and \{key=(1.1.1.1, 2.2.2.2, 123, 80, TCP), seq=3\}, the packet with $seq$=2 is lost. 
Our objective is to identify the flows with the largest fraction of missing packets, given a stream of packets on a network. When measuring the fraction of packets, we use the metric of packet byte count.

\textit{Packet drops to detect faulty links:} Faulty links can cause random packet drops along the network paths. When we detect the flows with high packet loss as in Section~\ref{sec:app_loss}, we can use the paths of these flows to identify potentially faulty links. For instance, when we detect both routed connections $A\to B\to C_1\to D$ and $A\to B\to C_2\to D$ have high packet loss, link $A\to B$ has a higher probability of being faulty. 

\para{Out-of-order packets:}
Given a stream of packets in one direction, the out-of-order packets are defined to be the packets whose sequence number are less than the current largest sequence number ($MaxSeq$), i.e., $seq < MaxSeq$, but arrives within a small period of time (e.g., 3ms) after the packet with $MaxSeq$ is received. 
Then the objective is to return the $k$ flows with the most out-of-order packets.

\textit{Out-of-order packets to infer incorrectly configured Quality of Service (QoS):} Although there are quite a few root causes of out-of-order packets, faulty QoS configurations, such as setting duplicate Access Control List (ACL) rules that misclassify the packets from the same application into a different application, 
can potentially delay some packets and fails to keep the packets in order. Such incorrect configuration can affect the performance of online applications. Thus, network flows with many out-of-order packets can be an indicator of present misconfigurations. As we describe in Section~\ref{sec:app_oop}, we detect flows with high out-of-order packets and use them for troubleshooting the network configurations.

\para{Retransmissions:} 
Given a stream of packets, the retransmitted packets are the packets whose sequence numbers appear at least twice in a flow. In TCP, retransmissions can happen under network congestion with high latency and packet loss. 
For instance, in TCP fast retransmission, duplicate ACKs are used as a part of packet recovery in order to remedy a packet loss or timeout. However, measuring only latency and packet loss may not be enough as retransmissions could happen due to other reasons, such as a buffer overflow in a video stream application in the user machine. Thus, measuring these retransmitted packets is useful in evaluating the performance of the prevailing connections.


In general, performance monitoring assists network diagnosis tools~\cite{Dhamdhere:2007,pingpoint} to further investigate in specific locations the root causes of network performance issues.

\section{Overview of Our Approach}
\label{sec:model}
In this section, we show a viable path towards lean algorithms for performance monitoring in the following steps:
\begin{itemize}
    \item[(a)] We describe a working model of computation for performance monitoring tasks.
    \item[(b)] We observe that a range of performance monitoring functions in the model share a similar property, which we characterize and call it flow-additive.
    \item[(c)] We show a characterization of flow-additive functions that can be estimated by lean algorithms. 
\end{itemize}
Following the above path, we then propose lean algorithms in the next section. We summarize our algorithms in Table~\ref{tab:summary_alg} and present the details in \S\ref{sec:algo}.

\para{Computation model and problem setting:} 
We model the network traffic as a massive data stream that is a union of flows. 
We measure some performance metrics defined over each flow and identify a subset of flows based on the statistics, e.g., flows with high latency and high packet loss. 
Prior work collects per-flow information to deterministically compute these performance metrics, which require $\O{N}$ memory for $N$ flows. 
This memory requirement limits the ability to deploy prior solutions on resource-constrained NIC and switch hardware as the number of flows can be large in practice (e.g., tens of millions in a public Internet trace~\cite{maccdc}).

Our goal is to design \emph{lean} performance-monitoring algorithms that use memory (and also implicitly processing time) sublinear in both the total size of the input traffic and the number of flows. 
Specifically, if $N$ is the number of distinct flows, then lean algorithms must use $o(N)$ memory. 
Thus, a natural question to ask is whether lean algorithms are even achievable for performance-monitoring metrics? 
We explore this question in the following subsection; we define the notions of heavy-hitters and flow-additive functions.

\para{$g$-heavy hitters and flow-additive functions:}
Suppose each network flow has been assigned an identity (e.g., 5-tuple), as previously discussed. 
Suppose $g(\cdot)$ is a function that takes arbitrary data as input and outputs a performance metric over the data. 
For instance, $g(\cdot)$ can be a function that counts packet loss so that $g(f_i)$ measures the packet loss in flow $f_i$. 
We also define $g(\cdot)$ over the union of all flows as $g\left(\cup f_i\right)$, where we use $\cup f_i$ to denote all packets that are sent across the network, i.e., $g\left(\cup f_i\right)$ is the aggregated packet loss among all flows. 
Then we would like to identify the most influential flows---the so-called $g-$heavy hitters that are ``heavy hitters'' among all flows computed by function $g$: 

\begin{definition}
For flows $f_1, f_2,\ldots, f_N$ and a function $g:f_i\to\mathbb{R}$, we say that a flow $f_i$ is $g$-heavy hitter if $g(f_i)$ equals at least some fraction $\eps$ of the $\ell_1$ or $\ell_2$ norm of the vector $(g(f_1),g(f_2),\ldots,g(f_N))$. 
Namely, the $\ell_1$ $g$-heavy hitters are those flows with value at least $\eps \cdot \sum_i |g(f_i)|$. Likewise, the $\ell_2$ $g$-heavy hitters are those with value at least $\eps \cdot \sqrt{\sum_i g(f_i)^2}$. 
\end{definition}

\begin{figure}[t]
\centering
\includegraphics[width=\linewidth]{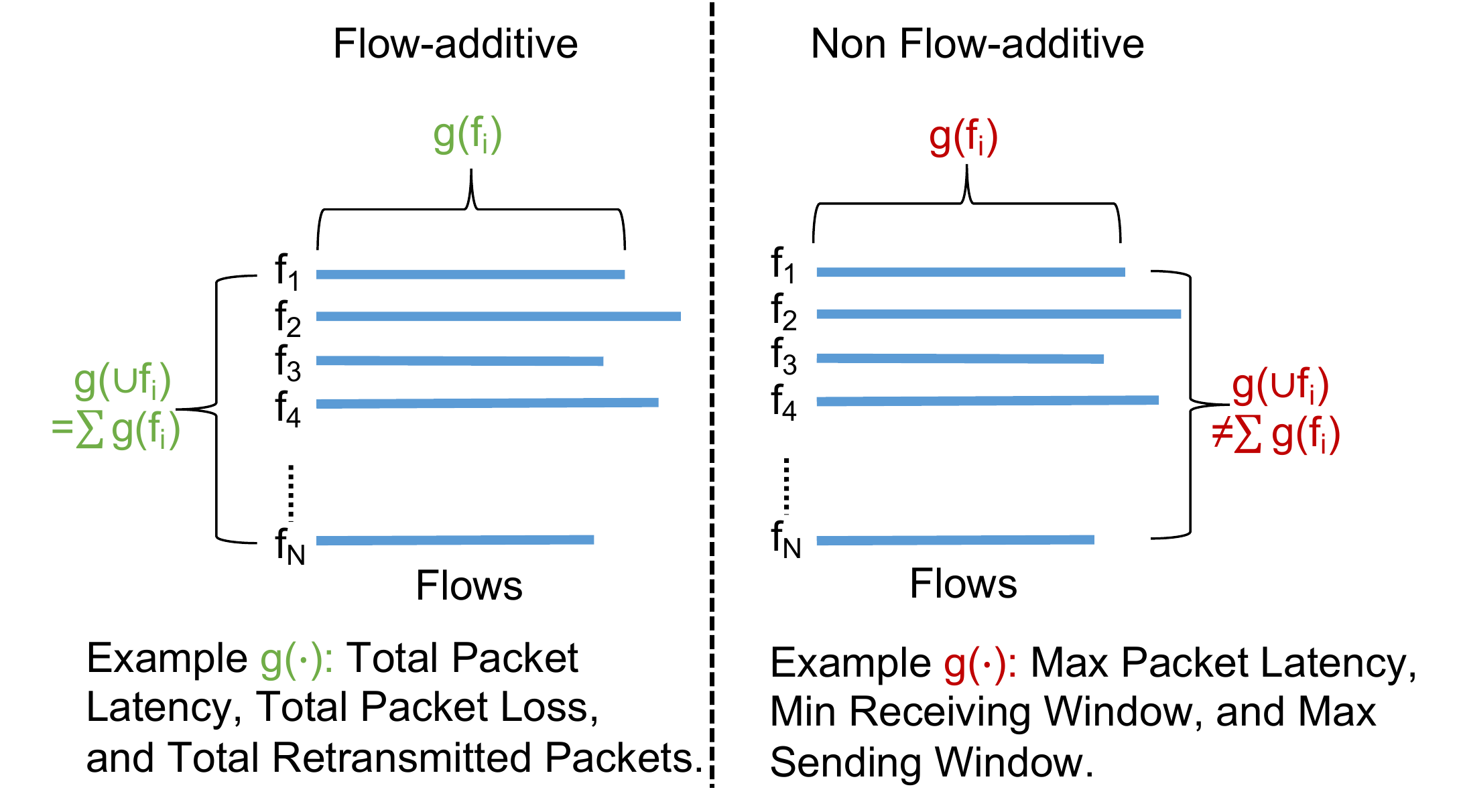}
\vspace{-3mm}
\tightcaption{Computation model for performance monitoring and flow-additive property.}\label{fig:lemma}
\end{figure}

\begin{table*}[t]
\centering
\footnotesize
\begin{tabular}{ llclc }
\toprule
\textbf{Perf. Stats $g(\cdot)$} & g-Heavy Hitters & Flow-additive / & Assumptions & Space /\\
 &  & Single Flow Sublinear &  & Time \\
\midrule
Latency & Flows with high latency & $\checkmark$ / $\checkmark$ & None & $\O{\frac{1}{\epsilon^2}\log n\log\frac{1}{\delta}}$\\
Loss & Flows with high packet loss & $\checkmark$ / X & Random packet loss & $\O{\frac{1}{\epsilon^2}\log n\log\frac{1}{\delta}}$\\
Out-of-order& Flows with high out-of-order packets &$\checkmark$ / X & Bounds packets received within 3ms & $\O{\frac{1}{\eps}\log n}$\\
Retransmittion & Flows with high packet retransmission & $\checkmark$ / X & Only on heavy flows & $\O{\frac{1}{\epsilon^2}\log n \log\frac{1}{\delta}}$\\
\bottomrule
\end{tabular}
\vspace{-1mm}
\caption{Summary of our lean algorithms for different performance statistics.}
\vspace{-3mm}
\label{tab:summary_alg}
\end{table*}

We first observe that a large number of performance monitoring functions, including total round trip time, out-of-order packets, packet retransmissions, and packet loss, obey a similar set of properties. 
Thus, it seems natural to investigate a unifying set of properties for which we might be able to characterize performance monitoring using lean algorithms. 
We formulate the desirable properties below and call functions that contain these properties \emph{flow-additive} functions.

Suppose that we have $N$ flows $f_1,f_2,\ldots,f_N$, as shown in Figure~\ref{fig:lemma}.
Each flow is a stream of packets associated with a packet payload, along with some other metadata, such as 5-tuple, packet type, sequence number, and timestamp. 
Given a function $g:f_i\to\mathbb{R}$, our goal is to approximate the $g$-heavy hitters\footnote{We give a proof to find the $\ell_2$ heavy-hitters, but observe that Theorem~\ref{thm:sublinear:exist} also reports the $\ell_1$ heavy-hitters and can further be modified to identify the $\ell_1$ heavy-hitters using space $\O{\frac{1}{\eps^2}s(g,N)\log N}$, such as in Algorithm~\ref{alg:rtt}.}, using both space and processing time that is sublinear in $N$.
As we will show in Theorem~\ref{thm:sublinear:exist}, to identify the $g$-heavy hitters, we require any function $g$ to be \emph{flow-additive}, which means that $g\left(\cup f_i\right)=\sum g(f_i)$.  
Intuitively, a flow-additive function on the entire network is the sum of the values of the functions for each flow in the network. 
For example, the function $g$ that computes total latency is flow-additive since the total latency of a network is defined to be the sum of the latencies of each flow. 
Other examples of flow-additive functions include total number of packets, total size of packets, number of out-of-order packets, number of retransmitted packets. 
On the other hand, the longest round-trip time is not flow-additive, since it is the maximum of the round-trip times of each flow.

Since a number of flow-additive functions seem to have natural lean algorithms to monitor the $g$-heavy hitters, we ask
\begin{quote}
For which flow-additive functions $g$ do there exist lean algorithms that compute the $g$-heavy hitters? 
\end{quote}

\begin{figure}[t]
\centering
\includegraphics[width=0.6\linewidth]{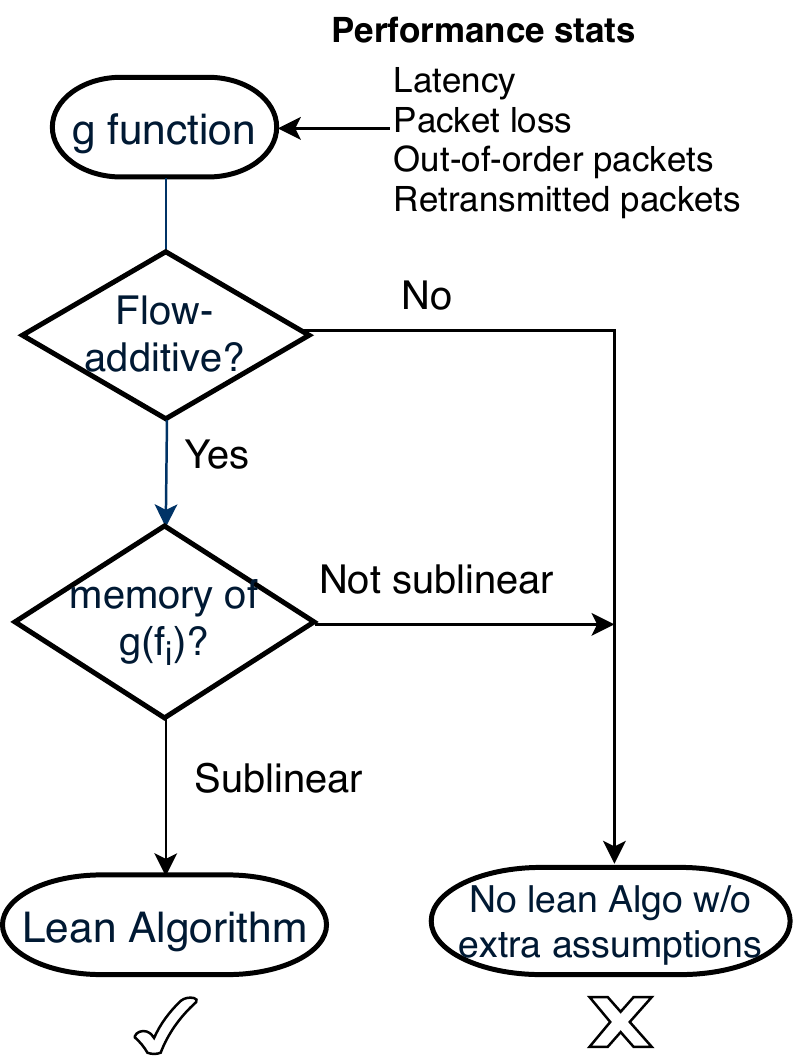}
\vspace{-1mm}
\tightcaption{High-level overview of Theorem~\ref{thm:sublinear:exist}.}\label{fig:lemma2}
\vspace{-2mm}\label{fig:model}
\end{figure}

\para{Characterization of flow-additive functions by lean algorithms:}
As shown in Figure~\ref{fig:lemma2}, we show that informally, a lean algorithm to compute the $g$-heavy hitters of a flow-additive function $g$ exists if and only if the memory consumption $s(g,M)$ of $g$ on a single flow is roughly sublinear (see Theorem~\ref{thm:sublinear:exist}) in $M$, the size of the input data \emph{for a particular flow}.

Intuitively, we observe that if a performance monitoring function is not flow-additive, then a significant challenge is finding the $g$-heavy hitters without maintaining approximate statistics for each flow; and if a function cannot even be estimated using space sublinear in the input data for a particular flow, then there is little hope it can be estimated across all of the flows. 
We show that some network performance issues, such as high latency, follow these two requirements and hence can be monitored by lean algorithms. 
At a high level, we randomly partition the input into several ``buckets'', so that all packets associated with flow $i$ are sent to the bucket associated with flow $i$. 
Given enough buckets, the $g$-heavy hitters are partitioned to separate buckets with high probability. 
We then run a separate algorithm for each bucket to identify each $g$-heavy hitter.

On the other hand, if some performance monitoring problems do not meet both of the two requirements above, such as  packet retransmissions, we can still find the $g$-heavy hitters with additional assumptions (see \S\ref{sec:app_loss}, for example). 

\begin{thm}
\label{thm:sublinear:exist}
Let $g$ be a flow-additive function, $\eps$ be the threshold for $g$-heavy hitters, and $N$ be the total number of flows, as well as an upper bound on the total input size per flow. 
There exists an algorithm to compute the $g$-heavy hitters using space $\O{s(g,N) \cdot \log N / \eps^4}$. 
Thus there exists a lean algorithm to compute the $g$-heavy hitters if $s(g,N)=o\left(\frac{N}{\log N}\right)$ and no such lean algorithm exists if $s(g,N)=\Omega(N)$. 
Moreover, if $s(g,N)=\polylog(N)$, then there exists a lean algorithm to compute the $g$-heavy hitters even if the total input size per flow is bounded by some polynomial in $N$.
\end{thm}
\begin{proof}
We first claim it is sufficient to consider the case where only a single flow $f_i$ has an $\ell_2$ heavy-hitter $g(f_i)$ (i.e., a flow $f_i$ whose value $g(f_i)$ is at least $\eps$ fraction of the total value $\sqrt{\sum g(f_i)^2}$ on the network). 
Note that at most $\O{\frac{1}{\eps^2}}$ flows can be $\ell_2$ heavy-hitters since $\frac{1}{\eps^2}(\eps\sqrt{\sum g(f_i)^2})^2=\sum g(f_i)^2$. 
Thus, by birthday bounds, hashing each flow into $\O{\frac{1}{\eps^4}}$ buckets maps each heavy-hitter to a separate bucket with constant probability. 
Hence, we can assume without loss of generality that only a single flow is a heavy-hitter in each bucket and by similar reasoning, we can assume that the contribution of the heavy-hitter to the bucket is greater than the sum of the remaining contributions to the bucket from other flows. 
Specifically, given a hash function $H:[N]\to [M]$ for $M = \O{\frac{1}{\eps^4}}$, we assume there exists some index $i$ such that

\[g(f_i)\ge\sum_{j:j\ne i,H(j)=H(i)}g(f_j).\]

We further partition each bucket into a number of sub-buckets, so that each sub-bucket receives all updates to a particular subset of all flows mapped to the bucket. 
We identify the index $i$ of the heavy-hitter mapped to this bucket in a bit-by-bit fashion by performing analysis on the sub-buckets as follows:
\begin{enumerate}
\item 
For each index $k$ with $1\le k\le\log N$, flow $j$ is mapped to sub-bucket $2k-1$ if the $k$-{th} bit of $j$ is $0$ and mapped to sub-bucket $2k$ if the $k$-{th} bit of $j$ is $1$.
\item
If $g$ computed on the sub-bucket $2k-1$ is greater than $g$ computed on the sub-bucket $2k$, we set the $k$-th bit of $i$ to be $0$. 
Otherwise, we set the $k$-th bit of $i$ to be $1$. 
\end{enumerate}

Observe that the first step partitions the flows of $\{j:j\ne i,H(j)=H(i)\}$ into $2\log N$ sub-buckets based on the parity of each bit of the index. 
However, since $g(f_i)\ge\sum_{j:j\ne i,H(j)=H(i)}g(f_j)$, then $g$ computed on the sub-bucket $2k-1$ has greater value than $g$ computed on the sub-bucket $2k$ if and only if the $k$-th bit of $i$ is $0$. 
Thus, the second step compares the values of $g$ computed on each pair of sub-buckets $2k-1$ and $2k$ to reveal the identity of the $k$-th bit of $i$, and so using all $\log N$ pairs of sub-buckets, we can identify all bits of $i$ (and thus $i$ itself). 
Since we use $\O{\frac{1}{\eps^4}\log N}$ sub-buckets, each requiring $s(g,N)$ space, the total space is $\O{\frac{s(g,N)\log N}{\eps^4}}$.

On the other hand, if approximating $g$ on a single flow with input size $N$ requires space $\Omega(N)$, then $g$ clearly cannot be monitored by a lean algorithm.
\end{proof} 

To show detailed examples of flow-additive functions and their corresponding lean algorithms, we design sketching algorithms to measure latency and packet loss, as described in~\S\ref{sec:app_latency} and~\S\ref{sec:app_loss}.

\section{Case Study with Lean Algorithms}\label{sec:algo}
In this section, we detail the specific algorithms for the case studies for which the computation model introduced in \S\ref{sec:model} can be monitored using lean algorithms. For simplicity of the presentation, we use Origin-Destination pair (OD-flow with a send and a receiver) to describe a network flow instead of using 5-tuple. Our algorithms estimate the 5-tuple flows in a similar way.

\subsection{Case Study I: High Latency}\label{sec:app_latency}
We formalize the problem as follows:
\begin{definition}[Latency Problem]
Given a stream of packets $e_i$ sender $s_i$, receiver $r_i$, packet id $p_i$ along with packet type $c_i$ ($c_i=-1$ if the packet is regular and $c_i=1$ if the packet is an acknowledgment), and time $t_i\in\mathbb{Z}^+$ at which the packet is recognized by the measurement point. Let the round trip time for packet id $p_i$ be $t_j-t_i$ where $t_j$ corresponds to the acknowledgment packet and $t_i$ corresponds to the regular packet.
Let $\RTT$ be the total round trip time for all packets in the data stream. When we measure only on certain types of packets, $\RTT$ is the aggregated round trip time for these packets.

Given constants $0<\eps,\delta<1$, we would like to output, with probability at least $1-\delta$, all flows $(s_i,r_i)$ whose total round trip time
exceeds $\eps\cdot\RTT$.
\end{definition}

A straightforward but na\"{\i}ve approach to monitor network latency would be to measure the latency across each of the $N$ flows in a large network, by storing per-flow latency measurements. 
However, the space required to store information for all flows could be prohibitively large. 

\para{Approach Overview:} 
We first show that the problem of round trip time is flow-additive and then describe the corresponding sublinear algorithm for measuring total RTT in a single flow. 
By Theorem~\ref{thm:sublinear:exist}, it follows that there exists a lean algorithm for monitoring influential flows with respect to total RTT, which we then explicitly describe.  
Given a flow $f_i$ connecting sender $s_i$ and receiver $r_i$, let $g(f_i)$ be the sum of the timestamps of the incoming packets (from $r_i$ to $s_i$) minus the sum of the timestamps of the outgoing packets (from $s_i$ to $r_i$). 
Then $\sum g(f_i)$ is the sum of the round trip times of all packets, which is exactly $g(\cup f_i)$, so the round trip time is flow-additive.

Intuitively, Theorem~\ref{thm:sublinear:exist} then provides a reduction to the problem of finding the heavy hitters in a stream. 
The universe of flows is the set of all sender-receiver pairs $(s_i,r_i)$. 
Upon sending a packet at time $t_i$, the ``counter'' for the pair $(s_i,r_i)$ is decremented by $t_i$, so that receiving the packet (e.g., ACK) at time $t_j$, the counter is incremented by $t_j$ and the total addition is $t_j-t_i$, the round trip time. 
We then report the pairs whose round trip time is an $\eps$-fraction of the total round trip time across all flows. 
Note that the heavy-hitter algorithm may not differentiate between a few packets with significantly high round trip time or significantly many packets with low round trip time. 
In practice, we can also obtain the mean round trip time if the flow is a heavy flow.
\begin{algorithm}[t]
\caption{Detect flows with high RTT}
\label{alg:rtt}
\textbf{Input:} A stream of elements $e_i=(s_i,r_i,p_i,c_i,t_i)$ with $s_i\in[n]$, $r_i\in[n]$, $c_i=\pm1$, and $t_i\in\mathbb{Z}^+$. \Comment{(Extracted packet header fields)}\\
\textbf{Output:} A list of  pairs $(s_i,r_i)$ with packets whose round trip time exceeds $\eps\cdot\RTT$.
\begin{algorithmic}[1]
\State{Let $N=2\cdot\binom{n}{2}$ and $B=\lceil\frac{9}{\eps^2}\rceil$.}
\State{Let $H_i: [N]\rightarrow[B]$ be a family of hash functions.}
\State{Let $G: [N]\rightarrow\{-1,1\}$ be a pairwise independent hash function.}
\State{Let $T$ be a $R\times B$ table, where $R=\lceil\log\frac{1}{\delta}\rceil$.}
\For{each element $e_i$:}
\If{$s_i<r_i$}
\State{Let $k\in[N]$ be the integer corresponding $(s_i, r_i)$.}
\For{each $1\le j\le R$}
\State{Set $T[j][H_j(k)]=T[j][H_j(k)]+G(k) \cdot t_i$.}
\EndFor
\Else
\State{Let $k\in[N]$ be the integer corresponding $(r_i, s_i)$.}
\For{each $1\le j\le R$}
\State{Set $T[j][H_j(k)]=T[j][H_j(k)]-G(k) \cdot t_i$.}
\EndFor
\EndIf
\EndFor
\For{each $(s_i,r_i)$ where $s_i<r_i$}
\State{Let $k\in[N]$ be the integer corresponding $(s_i, r_i)$.}
\State{The round trip time is the median of $|T[j][H_j(k)]|$ across $1\le j\le R$.}
\EndFor
\State{Output the flows whose round trip time is at least $\eps\cdot\RTT$.}
\end{algorithmic}
\end{algorithm}

\para{Analysis:} 
Formally, let $N=2\cdot\binom{n}{2}$, where $n$ is the number of nodes in the network. 
Let $H_i: [N]\rightarrow[B]$ be a family of hash functions, where $B$ is some integer that represents the number of buckets. 

Let $G: [N]\rightarrow\{-1,1\}$ be a pairwise independent hash function. 
Let $T$ be a $R\times B$ table, where $R=\lceil\log\frac{1}{\delta}\rceil$.

Let $s_i$ be a sender id and $r_i$ be a receiver id. 
Let $k\in[N]$ be the integer corresponding to $(s_i, r_i)$ for a regular packet and $(r_i,s_i)$ for an acknowledgment. 
Similarly, let $c_i=1$ if the packet is an acknowledgment and $c_i=-1$ otherwise.

Upon seeing each element $e_i=(s_i,r_i,p_i,c_i,t_i)$, we add $G(k)\cdot t_i \cdot c_i$ to $T[j][H_j(k)]$ for each $1\le j\le R$. 
At the end of the stream, the counter in the cell $T[j][H_j(k)]$ equals the sum of the round trip time of the packets from $s_i$ to $r_i$ in expectation, since $k$ corresponds to the ordered pair $(s_i,r_i)$ or $(r_i,s_i)$.
Thus, we use the median of the absolute values of all $H_i(k)$ for concentration inequalities. 
The intuition is that $c_i$ controls the direction of the message, so that regular packets and acknowledgment packets have different signs. 
Since $G$ is a pairwise independent hash function, the expectation of $T[j][H_j(k)]$ is the round trip total time between $s_i$ and $r_i$.

We present Algorithm~\ref{alg:rtt}, which follows CountSketch \cite{countsketch} for a universe of size $N$. 
\begin{thm}
\label{rtt:ex}
The expected value of $|T[j][H_j(k)]|$ for each $1\le j\le R$ is the cumulative round trip time of all packets on flow $(s_i,r_i)$.
\end{thm}

\begin{proof}
Suppose for $s_i<r_i$, the first message on flow $(s_i,r_i)$ is from $s_i$ to $r_i$ at time $t_i$ and the return message from $r_i$ to $s_i$ is at time $t_j$. 
Suppose also that no other message is passed between $s_i$ and $r_i$.
Recall that each flow is associated to some integer between $1$ and $N$. 
Let $k\in[N]$ be the integer corresponding to $(s_i, r_i)$. 
For each $k'\in[N]$, let $Y_{k'}$ be a random variable with $Y_{k'}=0$ if $H_j(k')\neq H_j(k)$ and $Y_{k'}=1$ otherwise (if $H_j(k')= H_j(k)$). 
Intuitively, $Y_{k'}$ represents whether the information of flow $k'$ is hashed to the same bucket as the information of flow $k$. 
Similarly, let $c(k')=-1$ for an acknowledgment packet and $c(k')=1$ for a regular packet.
Then for each $1\le r\le R$, 
\[T[r][H_r(k)]=\sum_{k'\in[N]}c(k')\cdot G(k')\cdot t_i\cdot Y_{k'}.\] 
Since $G$ is a pairwise independent hash function mapping to $\{-1,1\}$, then $\EEx{G}=0$ for $k'\neq k$, then $\EEx{T[r][H_r(k)]}=c(k)\cdot G(k)\cdot t_i$. 
Since $G(k)=\pm1$ and $c(k)=1$ for a regular packet and $c(k)=-1$ for an acknowledgment packet, then $|\EEx{T[r][H_r(k)]}|=t_j-t_i$. 
Thus, the expected value of $|T[j][H_j(k)]|$ for each $1\le j\le R$ is the round trip time between $s_i$ and $r_i$.
The proofs for the other cases are symmetric.
\end{proof}

\begin{thm}
\label{rtt:var}
For $1\le j\le R$, the variance of $|T[j][H_j(k)]|$ is the sum of the squared round trip times across all flows.
\end{thm}
\begin{proof}
For each flow $k$ connecting nodes $i$ and $j$, let $\RTT(k)$ be the total round-trip time of the packets between $i$ and $j$. 
For a fixed pair of nodes $u$ and $v$, let $k$ denote the flow associated with $(u,v)$. 
Then the variance of $|T[j][H_j(k)]|$ is 
\[\EEx{\left(\sum_{k'\in[N]}c(k')\cdot G(k')\cdot t_i\cdot Y_{k'}\right)^2}-\RTT(u,v)^2.\] 
Again note that $\EEx{G}=0$ for $k'\neq k$, so 
\[\var\left(T[j][H_j(k)]\right)=\EEx{\sum_{k'\in[N]}\RTT(k')^2\cdot Y^2_{k'}}-\RTT(u,v)^2.\]  
Since $H_j$ is a hash function mapping to $[B]$, then $\EEx{Y_k'}=\frac{1}{B}$ for $k'\neq k$ and so 
$\var\left(T[j][H_j(k)]\right)\le\frac{1}{B}\sum_{k'\in[N]}\RTT(k')^2.$\end{proof}

\begin{thm}
Let $\RTT$ be the total round trip time of all packets in the network.
Algorithm~\ref{alg:rtt} outputs all pairs $(s_i,r_i)$ such that the round trip time is at least $\eps\cdot\RTT$, with probability $1-\delta$ and using $\O{\frac{1}{\epsilon^2}\log n\log\frac{1}{\delta}}$ space and update time.
\end{thm}
\begin{proof}
Consider a fixed flow $k$ and hash function $H_j$. 
By Chebyshev's inequality on Theorems~\ref{rtt:ex} and \ref{rtt:var}, for $B=\frac{9}{\eps^2}$,
\[\Pr\left[\Big||T[j][H_j(k)]|-\RTT(k)\Big|\ge\eps\sqrt{\sum_{k'}\RTT(k')^2}\right]\le\frac{1}{3}.\]
Thus Algorithm~\ref{alg:rtt} detects whether $\RTT(k)$ is at least $\eps$ fraction of the total round trip time.
Taking the median of $\O{\log\frac{1}{\delta}}$ parallel hash functions boosts the probability to $1-\delta$. 
Observe that each update modifies a single entry in each row of the table by a basic arithmetic operation. 
Hence, the update time is also $\O{\frac{1}{\epsilon^2}\log n\log\frac{1}{\delta}}$.
\end{proof}

\begin{algorithm}[t]
\caption{Detect flows with high packet loss}
\label{alg:lost}
\textbf{Input:} A stream of elements $e_i=(s_i,r_i,p_i)$ with $s_i,r_i\in[n]$.\\
\textbf{Output:} A list of all pairs $(s_i,r_i)$ with large packet loss.
\begin{algorithmic}[1]
\State{Let $N=2\cdot\binom{n}{2}$ and $B=\lceil\frac{9}{\eps^2}\rceil$.}
\State{Let $H_i: [N]\rightarrow[B]$ be a family of hash functions.}
\State{Let $G: \mathbb{Z}\rightarrow\{-1,1\}$ be a pairwise independent hash function.}
\State{Let $T$ be a $R\times B$ table, where $R>1$ is any constant integer.}
\For{each element $e_i$:}
\State{Let $k\in[N]$ be the integer corresponding $(s_i, r_i)$.}
\For{each $1\le j\le R$}
\If{$p_i$ is odd}:
\State{Set $T[j][H_j(k)]=T[j][H_j(k)]+G\left(\frac{p_i+1}{2}\right)$.}
\Else
\State{Set $T[j][H_j(k)]=T[j][H_j(k)]-G\left(\frac{p_i}{2}\right)$.}
\EndIf
\EndFor
\EndFor
\For{each $(s_i,r_i)$}
\State{Let $k\in[N]$ be the integer corresponding $(s_i, r_i)$.}
\State{Let $f_k$ be the median of $|T[j][H_j(k)]|$ across $1\le j\le R$.}
\If{$f_k>\eps\sum_{x\in[N]}f_x$}
\State{Output $(s_i,r_i)$.}
\EndIf
\EndFor
\end{algorithmic}
\end{algorithm}

\subsection{Case Study II: High Packet Loss}\label{sec:app_loss} 
We present the problem as follows:
\begin{definition}[Packet Loss Problem]
Given a stream of elements $e_i=(s_i,r_i,p_i)$ with sender id $s_i\in[n]$, receiver id $r_i\in[n]$, and packet id $p_i$, let the number of lost packets $f_j$ for flow $j\in\left[N\right]$ be the number of packets that never appear and whose ids are less than the maximum packet id for flow $j$.  
Let $P$ be the total number of lost packets across all flows in the data stream. 
Given that some (constant) number of the flows have high packet loss, identify these flows in expectation.
\end{definition}

\para{Approach Overview:} 
We describe our approach in Algorithm~\ref{alg:lost}. 
We solve the problem by offering a reduction to the problem of identifying a random walk. 
Ideally, we would like to transform the updates so that a flow with high packet loss will be a random walk, while a flow that does not have high packet loss will map to a walk with moderate amounts of structure. 
Again, the universe of flows is the set of all sender-receiver pairs $(s_i,r_i)$.
We pair packet ids so that receiving both $p_i$ and $p_{i+1}$ will cancel out for $(s_i,r_i)$. 
On the other hand, if packet $p_i$ arrives but $p_{i+1}$ does not, then the counter for (the hash of) position $(s_i,r_i)$ will change by one in a random direction. 
Thus, if a large number of packets for $(s_i,r_i)$ is missing, then the counter will experience a random walk. 
We then report the pairs with large counters. 
Observe that the number of lost packets in a flow must also be an $\eps$-fraction of the total number of lost packets. 
Otherwise, the heavy-hitters algorithm cannot differentiate between one flow with a medium number of lost packets or several flows with a small number of lost packets. Formally, we have the following result:
\begin{thm}
Let $\alpha_i$ be the total number of packets lost across each flow $i$ and $\TPL=\sum\sqrt{\alpha_i}$. 
Algorithm~\ref{alg:lost} outputs all sender-receiver pairs $(s_i,r_i)$ with number of missing packets (assuming a uniform distribution of lost packets) is at least $\eps\cdot\TPL$, in expectation and using $\O{\frac{1}{\epsilon^2}\log n\log\frac{1}{\delta}}$ space and update time. 
\end{thm}

\para{Analysis:} 
Formally, let $n$ be the number of nodes in the network and $G:[m]\rightarrow\{-1,1\}$ be a hash function, where $m$ is the number of elements in the stream. 
Upon receiving packet $p_i$, define $q_i=G\left(\frac{p_i+1}{2}\right)$ if $p_i$ is odd and $q_i=-G\left(\frac{p_i}{2}\right)$ otherwise. 
Note that for a stream with no packet loss, $\left|\sum q_i\right|\le 1$. 
On the other hand, a stream that has missing packets corresponds to a random walk with unit step sizes. 
In fact, a stream with $m$ missing packets corresponds to a random walk of length $m$, which has length $\sqrt{\frac{2m}{\pi}}$ in expectation. 
Thus, we can identify a flow with a large number of lost packets. 
However, we note that Algorithm~\ref{alg:lost} only holds in expectation and may not provide a good estimation in some packet loss situations.

\subsection{Case Study III: High Out-of-Order Packets}\label{sec:app_oop}
To track all flows with a high number of out-of-order packets, we need to compare each incoming packet against the maximum sequence number and latest timestamp of the flow it belongs to. 
Without knowing this per-flow information, a specific packet cannot be classified as out-of-order packets. 
Therefore, a lower bound of $\Omega(n)$ counters are needed for  $n$ flows. 
\begin{lem}\label{lem:algo3:lower}
Any algorithm that finds all flows with a high number of out-of-order packets must use $\Omega(n\log n)$ bits of space, where $n$ is the number of flows.
\end{lem}
\begin{proof}
Consider the distribution where $n^2$ packets are partitioned among $n$ flows, and each of them contain no out-of-order packets. 
Thus, the state of the packets looks like ($p_1$, $p_2$, ..., $p_n$), where each $p_i$ is the packet number of flow $i$, and the sum of the $p_i$'s is $n^2$. 
There are roughly $\binom{n^2}{n}$ such possibilities, which is $2^{\Omega(n\log n)}$.

If we use less than $\frac{1}{2}\binom{n^2}{n}$ bits of space, a counting argument shows that there are ``many'' pairs of states $X$ and $X'$ that are mapped to the same memory configuration. 
Hence, there exists some flow $f_i$ such that $X$ reports $q_1$ packets seen by $f_i$ while $X'$ reports $q_2$ packets seen by $f_i$ with $q_1\neq q_2$. 
Without loss of generality, suppose the packets on $f_i$ reported by $X$ are $\{1,\ldots,q_1\}$ while the packets reported on $f_i$ by $X'$ are $\{1,\ldots,q_2\}$. 
Suppose an additional packet arrives on flow $f_i$ and with $\frac{1}{2}$ probability, the packet ID is $q_1-1$ and with $\frac{1}{2}$ probability, the packet ID is $q_2-1$. 
Then the probability the algorithm correctly identifies the number of out-of-order packets is at most $\frac{1}{2}$. 
\end{proof}

\begin{algorithm}[t]
\caption{Algorithm for out-of-order packets}
\label{alg:oopacket}
\textbf{Input:} A stream of elements $e_i=(s_i,r_i,p_i, t_i)$.\\
\textbf{Output:} A list including all large flows $(s_i,r_i)$ with a high number of out of order packets.
\begin{algorithmic}[1]
\State{Let $T$ be a $r \times 2$ table, with $r=\frac{1}{\epsilon}$.}
\State{Let $Q$ be a min priority queue where each element is an ordered pair $(f_i,t_i)$ consisting of a flow $f_i$ and the time the last packet for the flow was received.}
\For{each element $e_i$:}
\State{Let $k_i\in[N]$ be the flow associated with $(s_i,r_i)$.} 
\If{$k_i$ is in $Q$}
\If{$p_i\le MaxSeq$}
\If{$T[j][1]=k_i$ for some $j$}
\State{$T[j][2]\gets T[j][2]+size(e_i)$.}\\
\Comment{Increment counter for $e_i$. When each packet is considered as the same weight, $size(e_i)=1$}
\ElsIf{$T[j][2]\ge size(e_i)$ for all $j$}
\State{$T[j][2]\gets T[j][2]-size(e_i)$ for all $i$.}\\
\Comment{Decrement all counters.}
\Else
\State{$z\gets\argmin T[j][2]$.}
\State{$y=size(e_i)-T[z][2]$.}
\State{$T[j][2]\gets T[j][2]-y$ for all $j$.}
\State{$T[z][1]\gets k_i$ and $T[z][2]=y$.}
\EndIf
\Else
\State{$MaxSeq\gets p_i$.}
\EndIf
\Else
\State{Q.push($(f_i,t_i)$)}
\EndIf
\For{any $j$ with $t_j<t_i+3ms$}
\State{Q.pop($(f_j,t_j)$)}
\EndFor
\EndFor
\end{algorithmic}
\end{algorithm}

Although it does not seem evident how to approximate the number of out-of-order packets on a single flow using space sublinear in the input of the flow, we nevertheless obtain a lean algorithm with the following relaxation. 
Namely, if we assume that the number of packets that arrive within some time window, such as 3ms, is a constant bounded amount, then we can track the number of out-of-order packets on a single flow using sublinear space. 
Consider the following variant of the out-of-order problem:

\begin{definition}[High Out-of-Order Packets]
Given a stream of elements $e_i=(s_i,r_i,p_i, t_i)$ with sender id $s_i\in[n]$, receiver id $r_i\in[n]$, packet id $p_i$, and timestamp $t_i\in[m]$, we define the number of out-of-order packets for flow $j\in[n]$ as the number of packets that are not received in order. 
Suppose $MaxSeq$ is the maximum received ID number for a flow so that out-of-order packets have $p_i < MaxSeq$ and further assume that out-of-order packets arrive within some period of time (e.g. 3ms) after the packet with $MaxSeq$ was transmitted. 
Let $P$ be the total number of out of order packets across all flows in the data stream. 
Given a constant $0<\eps<1$, we would like to output all flows $i$ such that the number of out of order packets $f_i$ exceeds $\eps\cdot P$.
\end{definition}

Thus, we have the following result and give in Algorithm~\ref{alg:oopacket} the full algorithm for the high out-of-order packets problem. 

\begin{thm}[Informal]
Let $P$ be the number of out-of-order packets (that appear within 3 ms of the latest packet). 
Algorithm~\ref{alg:oopacket} returns the flows that have at least $\eps\cdot P$ out-of-order packets, if any exist, using $\O{\frac{1}{\eps}\log n}$ space and update time, along with the space necessary to maintain all packets within that arrive within $3ms$.
\end{thm}

Here we track the flows with the highest number of out-of-order packets using a tug-of-war sketch. 
At each point, we maintain counters for certain flows that we have seen so far. 
When one of these flows is determined to have an additional out-of-order packet, the corresponding counter for that flow is incremented. 
When a different flow with an out-of-order packet is encountered, each counter is decremented. 
If the counter for a certain flow reaches zero, it can be replaced with a different flow. 
Since there are $\frac{1}{\eps}$ counters, any flow with at least $\eps\cdot P$ out-of-order packets will be output by the data structure at the end.

\subsection{Case Study IV: High Number of Retransmitted Packets}
\label{sec:app_retransmission}
Suppose there exists a traffic network with a central hub that can measure packets according to the route of the packet and the packet ID. 
One attribute of a problematic flow is a high number of retransmissions, packets that somehow fail to send and must be resubmitted. 
Although there does not seem to be an obvious way of approximating the number of retransmissions on a single flow using space sublinear in the input of that particular flow, we can still obtain a lean algorithm with the following relaxation. 
We formalize the problem as follows: 
\begin{definition}[Retransmitted Packets Problem]
Given a stream of elements $e_i=(s_i,r_i,p_i)$ with sender id $s_i\in[n]$, receiver id $r_i\in[n]$, and packet id $p_i$, let the number of retransmitted packets $f_j$ for flow $j\in\left[\binom{n}{2}\right]$ be the number of packets whose ids appear at least twice. 
We say that a flow has high transmission if the average packet in the flow is retransmitted $k$ times for some $k>1$. 
Given that some (constant) number of the flows has high retransmission, identify these flows with probability at least $2/3$.
\end{definition}
We instead relax the problem to finding the elephant flows that have high retransmission. 
That is, we report the flows that have high retransmission and at least $\eps$-fraction of $T$, the total number of packets sent across the network. 

We use a CountSketch algorithm to continuously report the flows with at least $\frac{\eps}{2} T$ packets. 
For each flow reported by the algorithm, we approximate the average number of retransmissions in the flow by tracking the total number of subsequent packets in the flow, as well as an approximation of the number of distinct packets sent by the flow, such as by using an algorithm of~\cite{KaneNW10}.   
If the CountSketch algorithm ever stops reporting that a particular flow is a heavy-hitter, then we discontinue tracking of the packets of the flow. 
Observe that any with at least $\eps T$ packets is tracked after $\frac{\eps}{2}T$ packets arrive. 
Hence if the flow sent each packet an average of at least $k$ times, then on the remaining packets, the average number of retransmissions for each packet is at least $\frac{k}{2}$, after the flow is tracked. 
Since we maintain a $2$-approximation of the number of distinct elements, then the average \emph{reported} number of retransmissions is at least $\frac{k}{4}$. 
On the other hand, if the flow previously sent each packet an average of less than $\frac{k}{16}$ times, then on the remaining packets, the average number of retransmissions for each packet is less than $\frac{k}{8}$, after the flow is tracked. 
Hence, the average reported number of retransmissions is less than $\frac{k}{4}$.

\begin{algorithm}[t]
\caption{Algorithm for high retransmissions}
\label{alg:retransmissions}
\textbf{Input:} A stream of elements $e_i$, threshold $k$.\\
\textbf{Output:} A list including all large flows with an average number of retransmissions at least $k$.
\begin{algorithmic}[1]
\State{Use CountSketch to maintain a list $L$ of the flows with at least $\frac{\epsilon}{2}$ fraction of the total packets.}
\For{each element $e_i$:}
\For{each flow $f_j$ in $L$:}
\State{Track the number of elements $n_j$ in $f_j$.}
\State{Maintain a $2$-approximation of the number of distinct elements $d_j$ in $f_j$.}
\EndFor
\EndFor
\State{$F\leftarrow\emptyset$}
\For{each flow $f_j$ in $L$:}
\If{$\frac{n_j}{d_j}\ge\frac{k}{4}$:}
\State{Append $f_j$ to $F$.}
\EndIf
\EndFor
\State{Return $F$.}
\end{algorithmic}
\end{algorithm}

We give the full algorithm for the (relaxed) high number of retransmitted packets problem in Algorithm~\ref{alg:retransmissions}.
Thus, we have the following result:

\begin{thm}[Informal]
Let $T$ be the total number of packets sent on the network. 
There exists a streaming algorithm outputs all flows that send at least $\eps T$ packets and have an average retransmission rate of at least $k$ and reports no flows that have an average retransmission rate of less than $\frac{k}{4}$. 
The algorithm succeeds with constant probability and uses $\O{\frac{1}{\eps^2}\log n}$ space and update time.
\end{thm}

\section{Implementation}\label{sec:impl}

We implemented prototypes of Algorithms 1 to 4, including a switch data-plane program (in P4-14) and a controller (in Python). In the data plane, we define each algorithm's per-packet behaviors through a series of processing stages, each of which has its dedicated resources, including register arrays and match-action tables (as shown in Figure~\ref{fig:pisa_overview}). 
Our data-plane program is compiled to Barefoot P4 studio with all algorithms combined. 

\begin{figure}[t]
\centering
\includegraphics[width=0.96\linewidth]{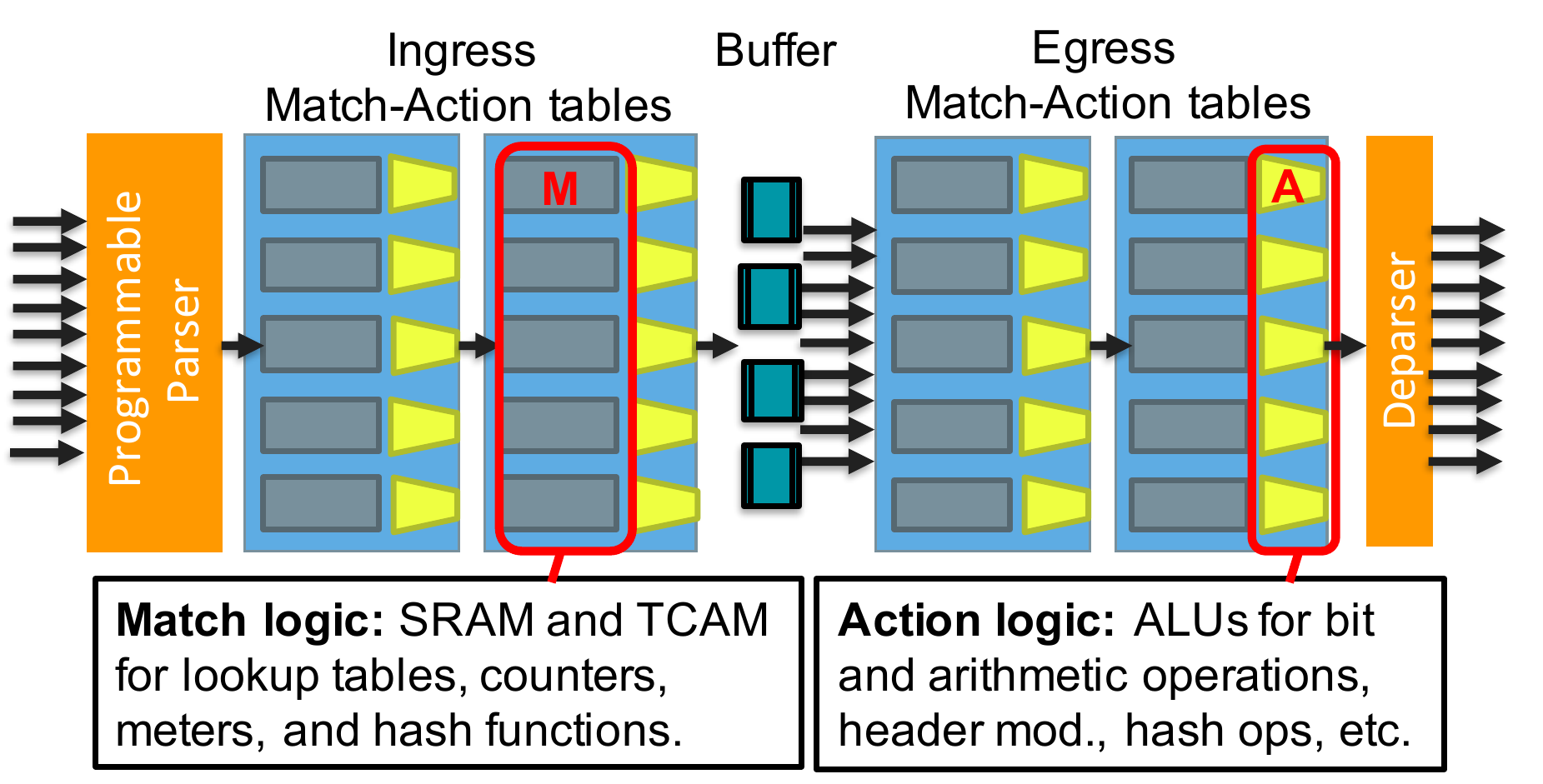}
\tightcaption{Programmable switch workflow.}
\label{fig:pisa_overview}
\end{figure}

\para{Data-Plane Implementation:} To realize our data-plane algorithms, there are three steps as the following: 
\begin{packeditemize}
\item[(1)] Extract required header fields (e.g., SrcIP, DstIP, Seq No., Timestamp, and Proto.) into P4 metadata with the programmable parser. These metadata are shared and can be accessed among all processing stages in a programmable switch.
\item[(2)] Leverage embedded CRC32 hash functions to hash the \emph{flow key} and update to the corresponding counters in register arrays with algorithm-specific updating schemes (e.g., plus or minus with current timestamp); 
\item[(3)] Report possible influential flow keys to the controller for offline estimation of the performance statistics.
\end{packeditemize}

\para{Practical considerations in the implementation:} We focus on a hardware switch implementation and have the following adjustments to account for hardware limitations. 
\begin{enumerate}[leftmargin=*]
    \item Hash functions: We use Count Sketch~\cite{countsketch} as a component in our algorithms, and the analysis of Count Sketch~\cite{countsketch} requires pair-wise independent hash functions. Barefoot Tofino switches have no hardware support for such guarantees. Thus, to ensure good ``independence'' between hash functions, we configure the embedded CRC32 with random polynomial hash seeds and select a set of hash seeds that produce significantly different hash values for the same flow key.
    \item Priority queue: In Algorithm 3, we use a min priority queue to maintain the set of most recent received packets. However, such a data structure with a non-linear number of operations per packet is not supported in existing programmable switch hardware. Instead, we use a two-way cuckoo table to cache the most recent packet information and we observe a negligible number of collisions when the size of the table is reasonably large (e.g., several times of receiving window size). 
    \item Top flow keys: When reporting the identities of the most influential flows, we cannot leverage data structures such as heap or priority queue to store the flow identities in the data plane. Instead, we leverage a packet mirroring feature in the switch with a Bloom filter~\cite{Bloom,  LuoSurvey19, FPFZ18}. Once the estimated statistic for a flow exceeds some threshold, the switch duplicates this packet and reports the copy to the controller. Since the number of possible influential flows is small, the Bloom filter in the switch (almost) ensures that each influential flow is reported once. An alternative implementation without any false positives is to use TCAM (ternary content-addressable memory) to maintain a one-to-one matching table. 
    \item Timestamp: When measuring the flow latency, we need to record the timestamps of each targeted pair of packets as the timestamp field in the TCP packet is not useful in our task. Once the programmable parser matches a particular type of packet, we leverage the high-precision timer in the switch to record \code{ingress_global_timestamp} in the P4 metadata. This \code{ingress_global_timestamp} comes with nanosecond-level precision.
\end{enumerate}
In \S\ref{sec:eval:hardware}, we detail the hardware resource usage of our prototype by compiling to Barefoot Tofino with Barefoot P4 Studio suite. 

\para{Control Plane:} We implement the controller as a Python module. The P4 framework allows us to define the API for control-data plane communication. We use a Thrift API to query the contents of data-plane register arrays (sketch data structure). After obtaining sketches, we estimate the top-$K$ influential flows by estimating the values of the stored flow identities on the sketches.
\begin{figure}[t]
\centering
\includegraphics[width=0.7\linewidth]{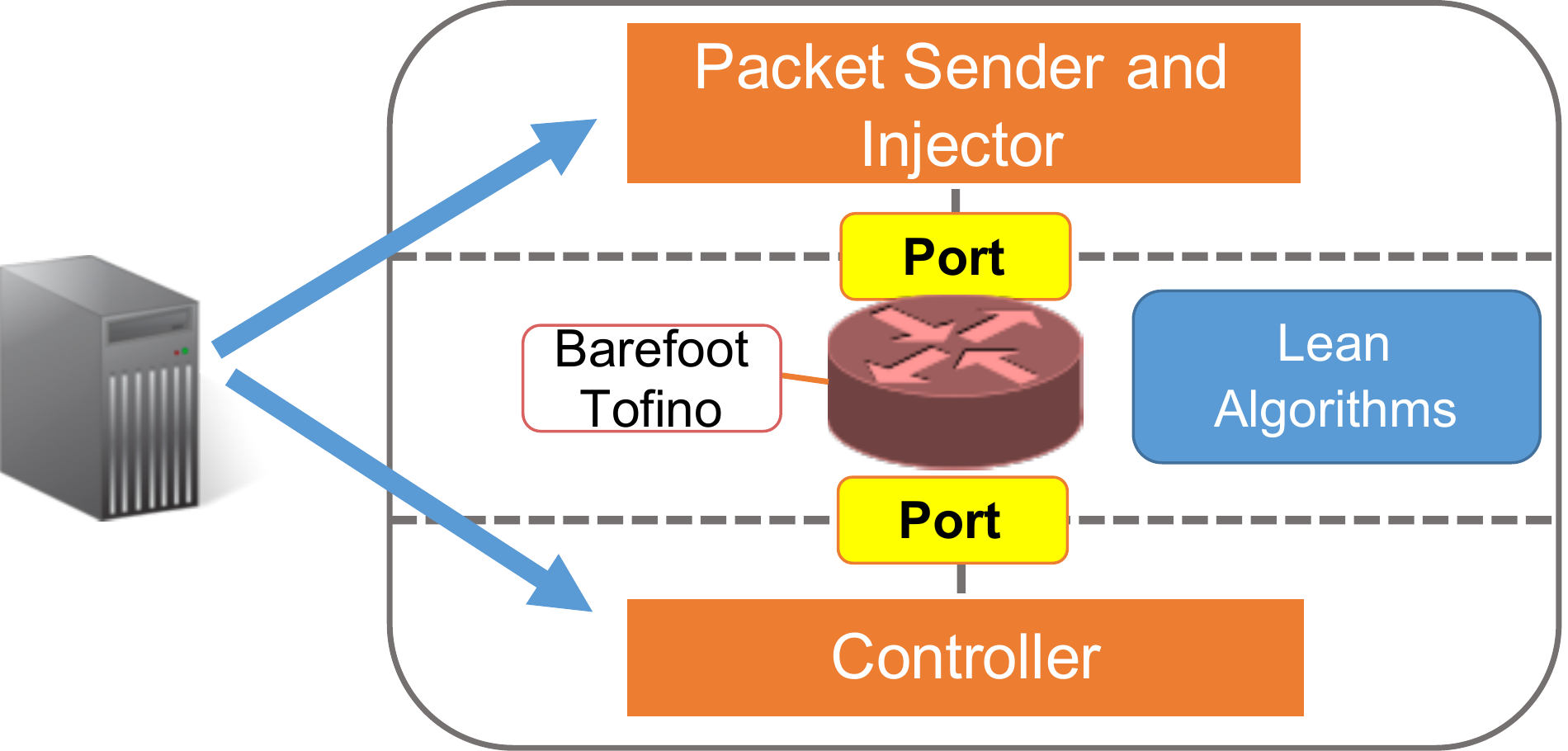}
\tightcaption{Evaluation
Setup.}
\label{fig:eval_overview}
\end{figure}

\begin{figure*}[t]
\centering
\subfigure[Latency]{
\label{fig:latency}
\includegraphics[width=0.237\textwidth]{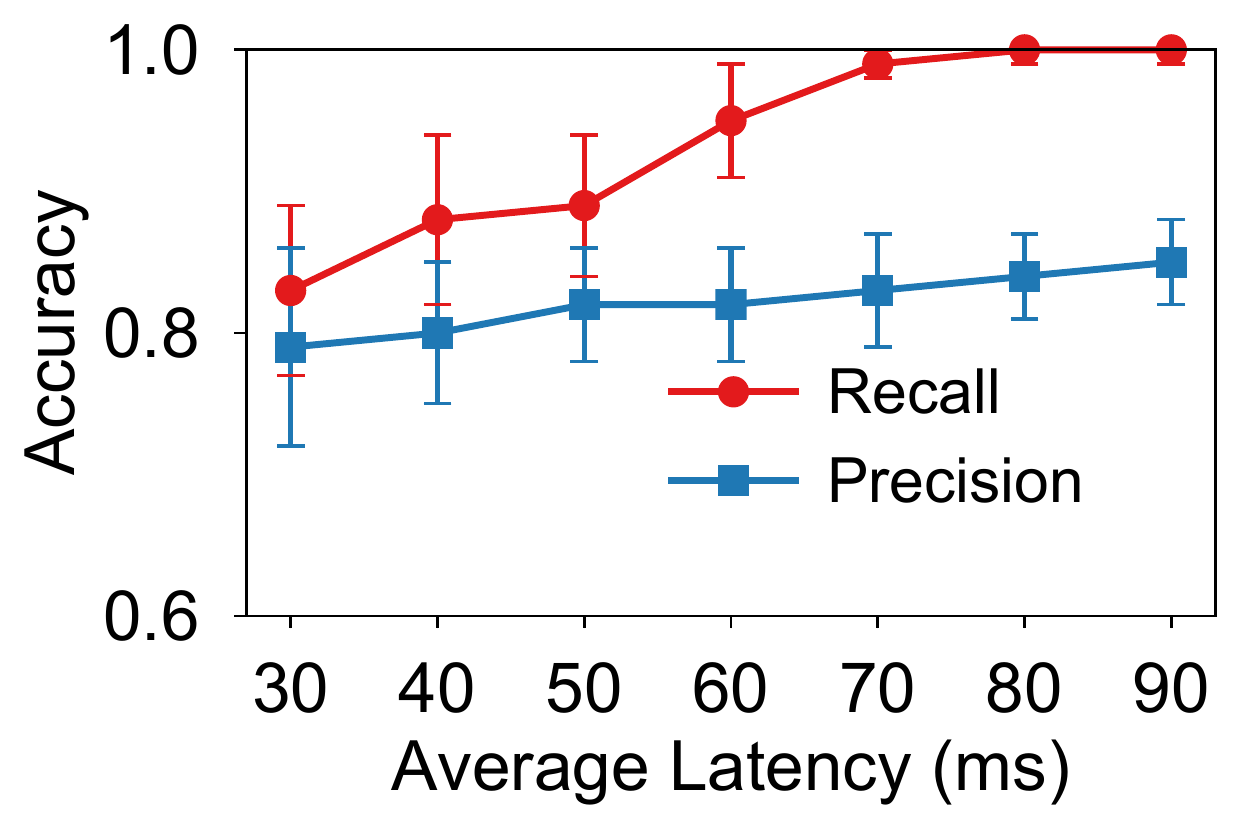}}
\subfigure[Latency]{
\label{fig:latency_tradeoff}
\includegraphics[width=0.237\textwidth]{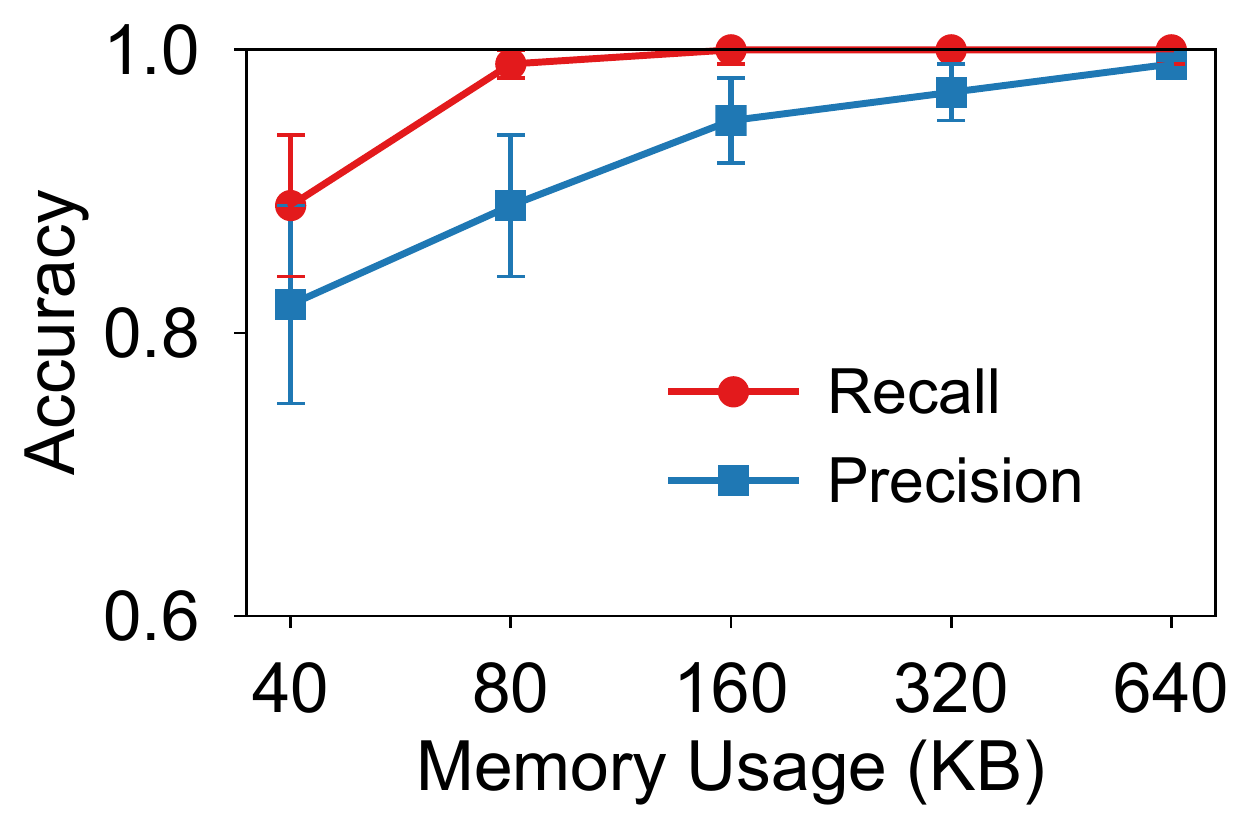}}
\subfigure[Packet Loss]{
\label{fig:packet_loss}
\includegraphics[width=0.237\textwidth]{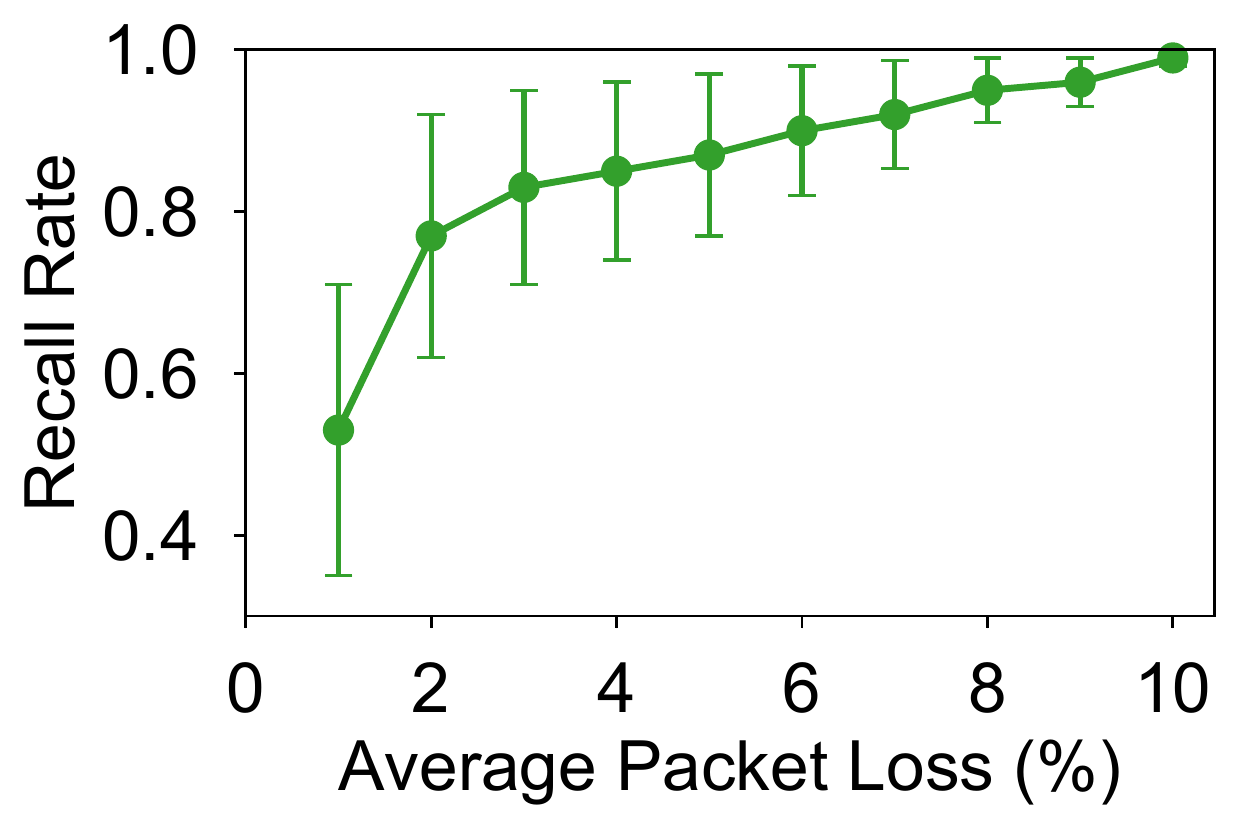}}
\subfigure[Packet Loss]{
\label{fig:packet_loss_tradeoff}
\includegraphics[width=0.237\textwidth]{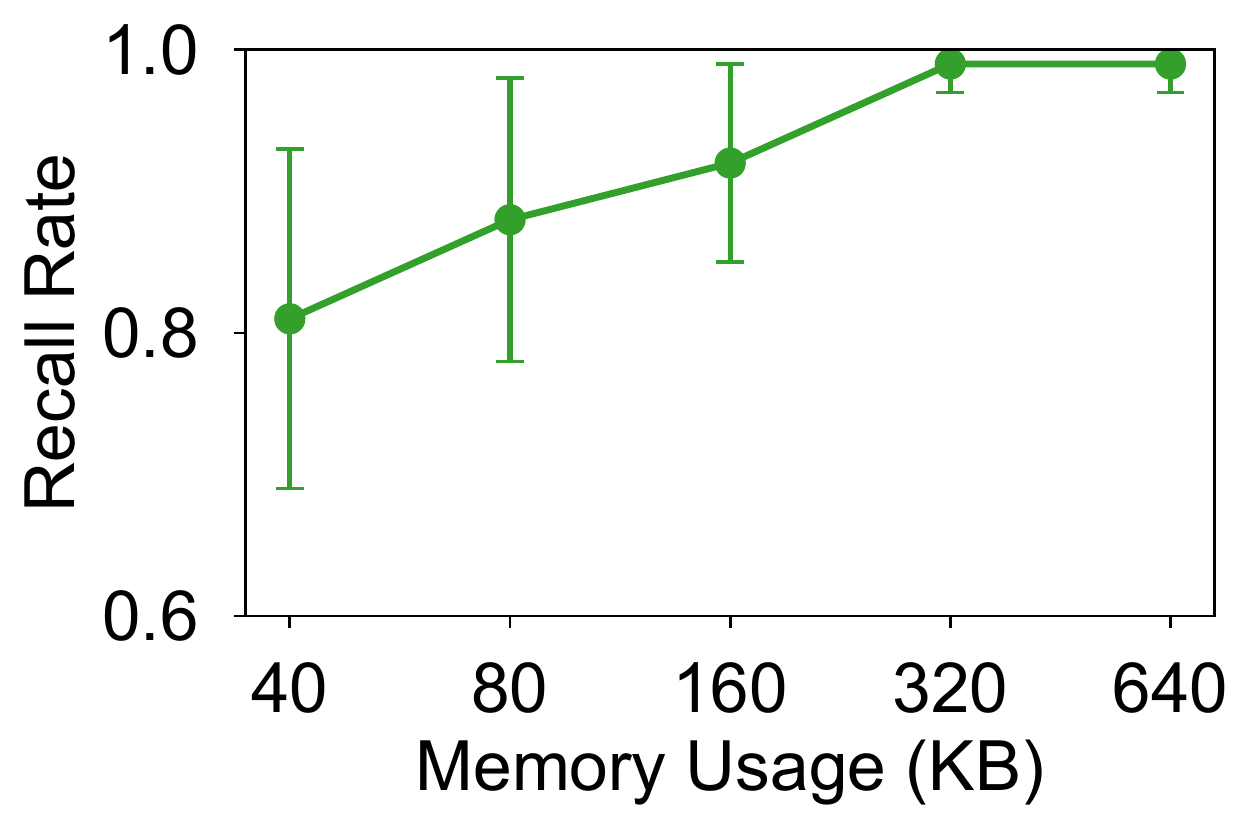}
}

\tightcaption{(a) Latency vs. accuracy (recall and precision) in detecting top 100 flows with various latency. (b) Memory vs. accuracy tradeoff in detecting top 100 flows with 50ms average latency. (c) Packet loss vs. recall in detecting top 100 flows with various packet losses. (d) Memory vs. recall tradeoff in detecting top 100 flows with 4\% packet loss.}
\end{figure*}


\begin{figure*}[t]
\centering
\subfigure[Out-of-order Packet]{
\label{fig:packet_outoforder}
\includegraphics[width=0.237\textwidth]{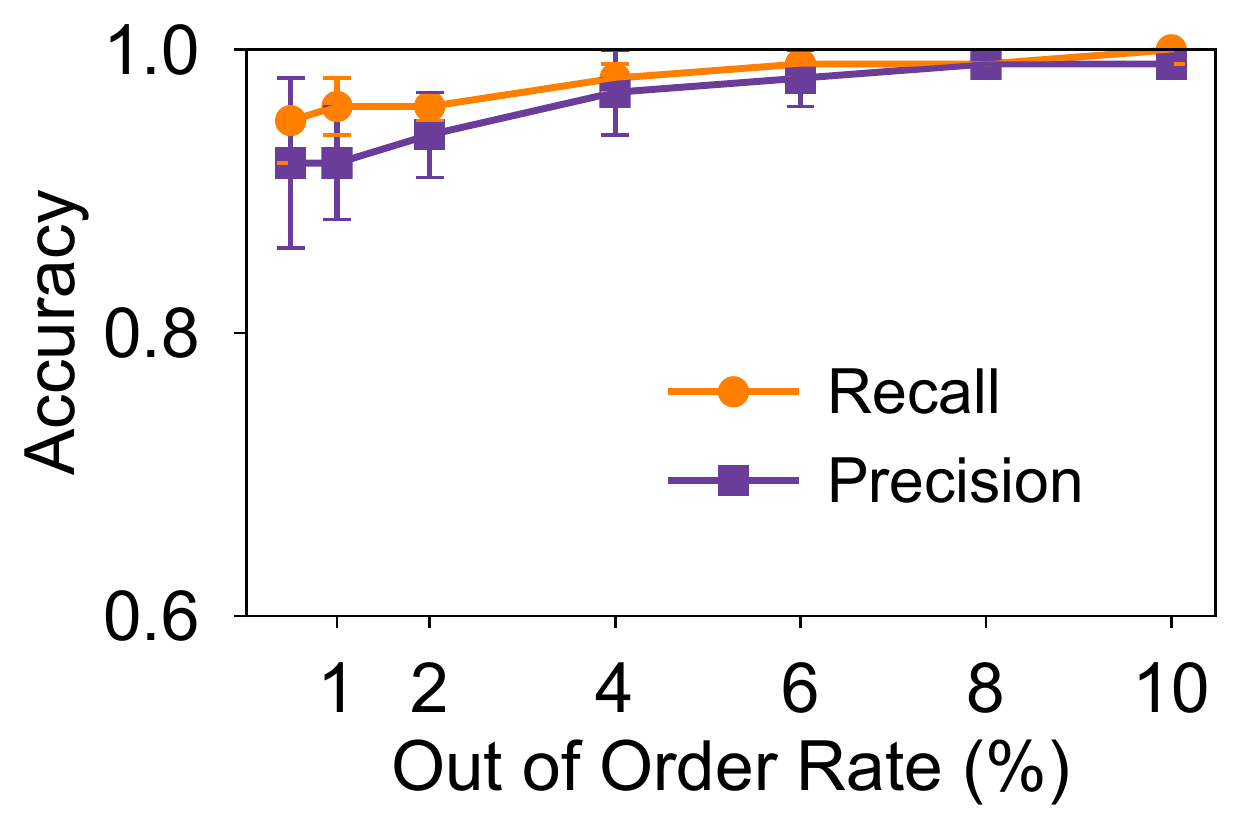}}
\subfigure[Out-of-order Packet]{
\label{fig:packet_outoforder_tradeoff}
\includegraphics[width=0.237\textwidth]{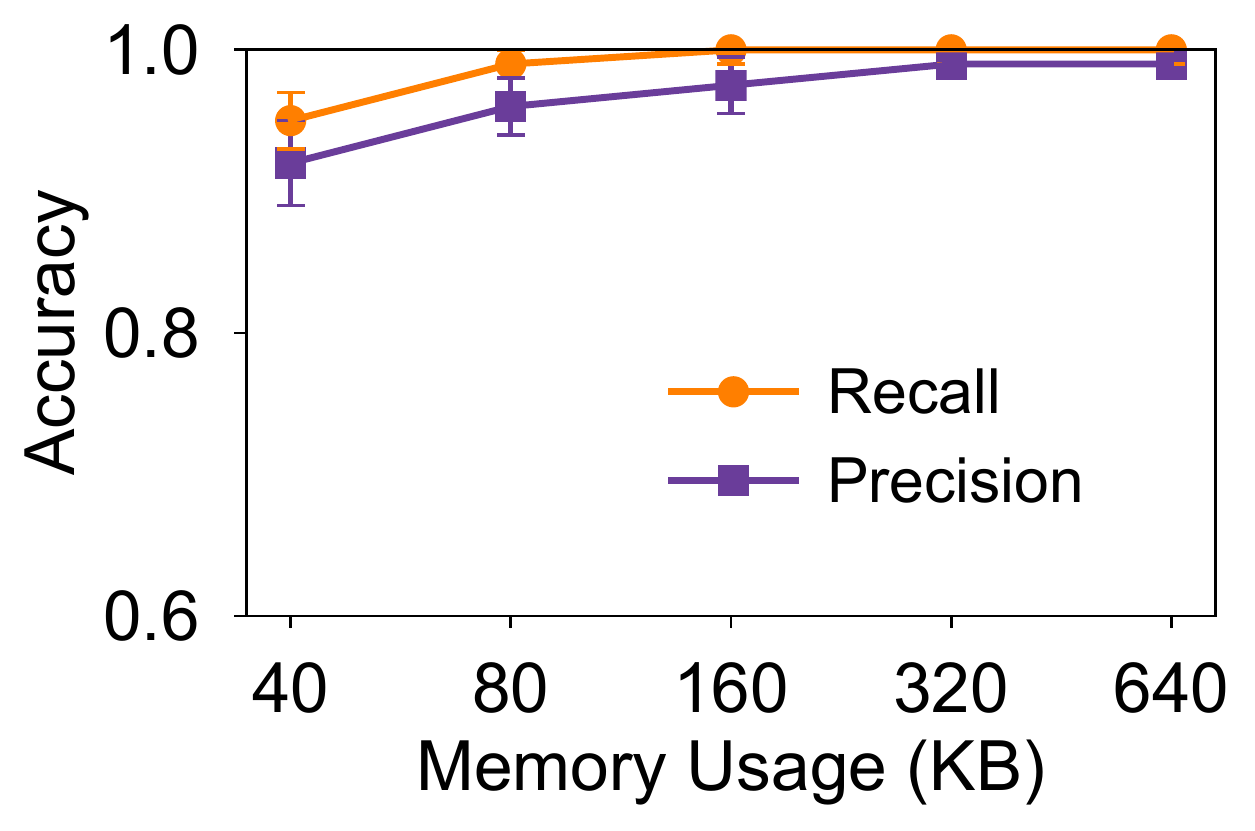}}
\subfigure[Retransmission]{
\label{fig:retransmission}
\includegraphics[width=0.237\textwidth]{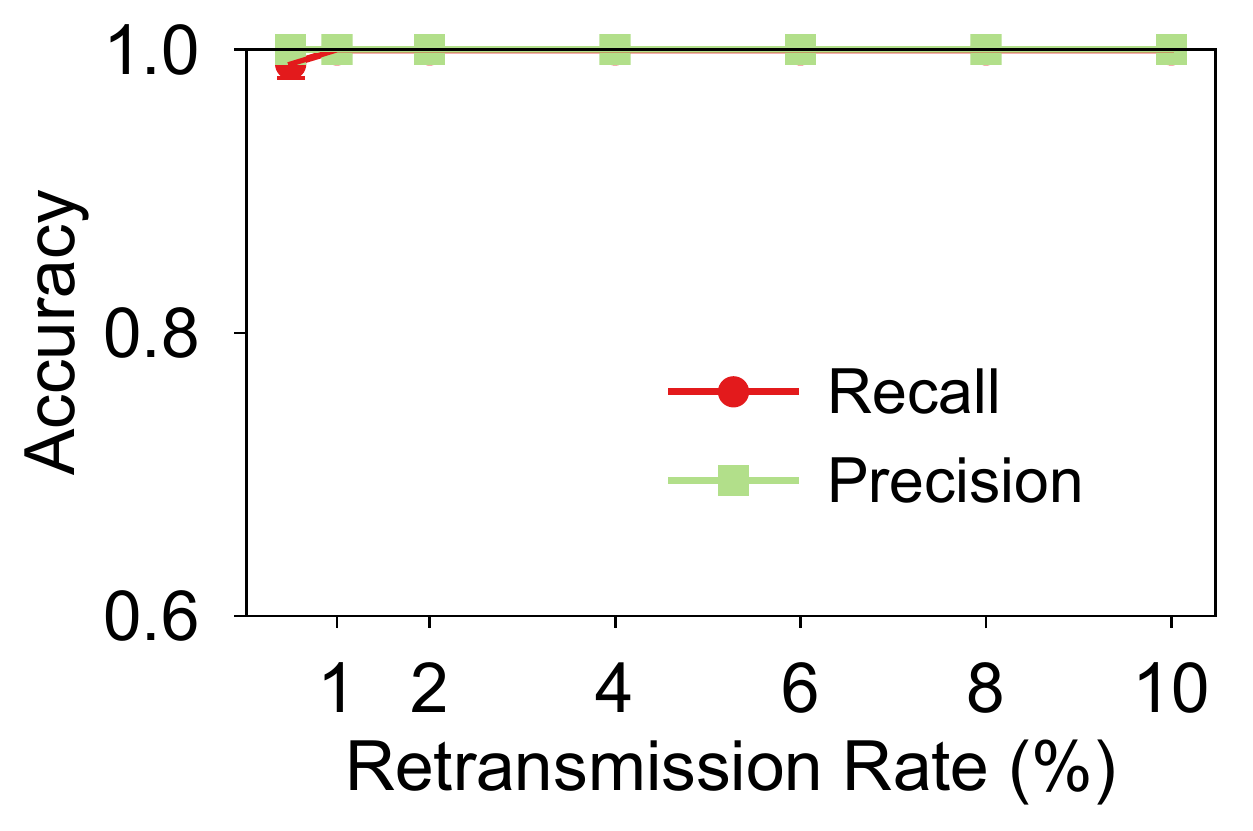}}
\subfigure[Retransmission]{
\label{fig:retransmission_tradeoff}
\includegraphics[width=0.237\textwidth]{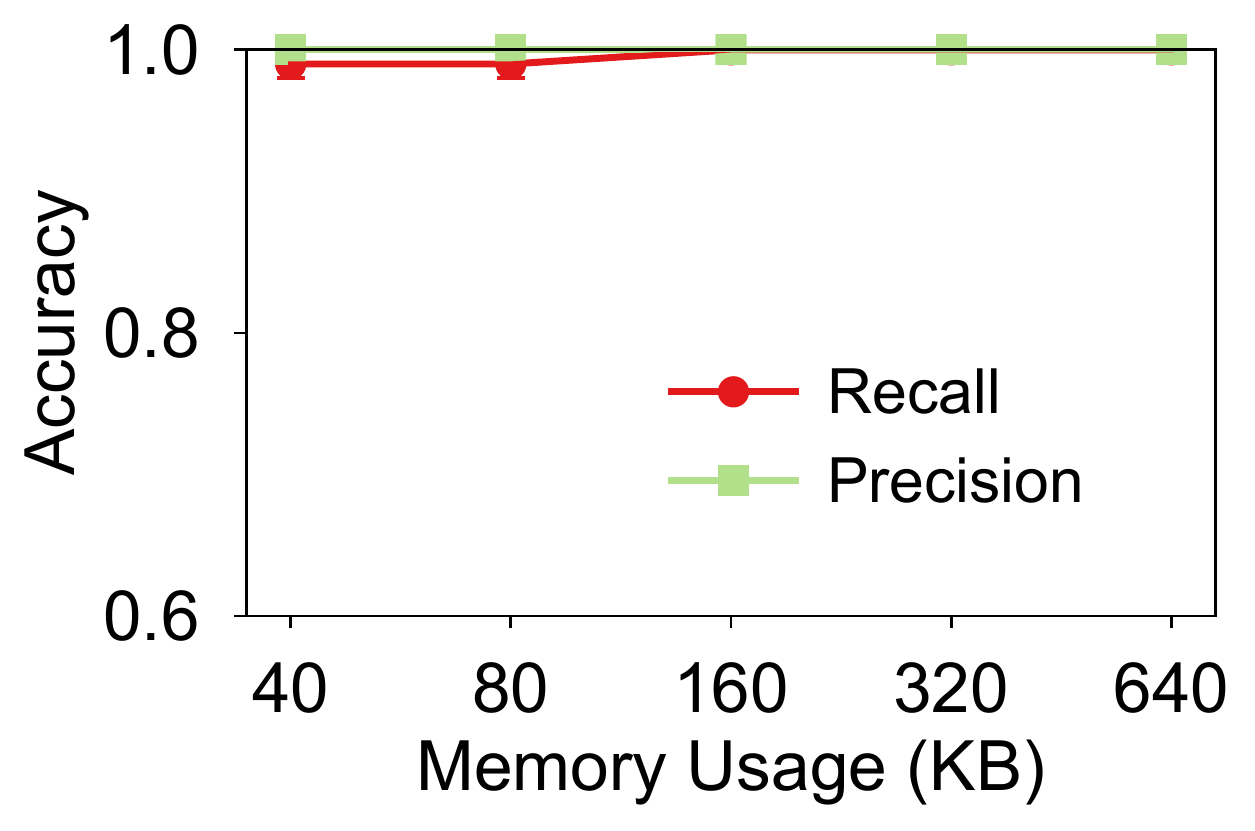}}

\tightcaption{(a) Out-of-order rate vs. accuracy (recall and precision) in detecting top 100 influential flows. (b) Memory vs. accuracy tradeoff in detecting top 100 flows with 4\% out-of-order rate. (c) Retransmission rate vs. accuracy in detecting top 100 flows. (d) Memory vs. accuracy tradeoff in detecting top 100 flows with 1\% retransmission rate.}
\end{figure*}

\section{Evaluation}\label{sec:eval}
In this section, we evaluate the performance of our algorithms using CAIDA's anonymous network traces from 2018~\cite{caida2018}. We show that our prototype can (1) report flows with high latency, high out-of-order packets, and high retransmission rates accurately, and achieve high recall rates when packet loss is significant; (2) present a tradeoff between memory vs.\ accuracy and offer reasonable accuracy even with small memory footprint; and (3) measure multiple statistics simultaneously under hardware resource limitations. We first briefly describe our methodology in Section~\ref{sec:eval:method} and show the results in the rest of the section. Other base-line solutions~\cite{snap,marple,dapper} provide $100\%$ accuracy but use $\O{N}$ space for $N$ flows.

\subsection{Methodology}\label{sec:eval:method}

We build our prototype with P4-14 on a 6.5Tbps Barefoot Tofino switch and set up the environment as shown in Figure~\ref{fig:eval_overview}. 
We use CAIDA anonymized traces collected from an OC-192 link between Sao Paulo and New York. The traces are divided into one-minute epochs, each of which consists of approximately 70 million packets. In our evaluation, we set our monitoring interval to 1-min, and each epoch has about approximately 3,700,000 unique flows. By default, we allocate 5 rows of 2000 32-bit counters (40KB) in our sketch.

\para{Setups:} As public packet traces do not have significant performance issues in their network, we create a synthetic packet trace from CAIDA traces: (a) To simulate the high-latency issue, we manually inject 30$\sim$90ms delays into 100 random heavy flows before sending to the switch with the packet sender. (2) To simulate high packet loss, we inject 1\%$-$10\% random packet loss to the randomly selected 100 flows out of the top 1000 heavy flows. (3) To simulate high out-of-order packets, we randomly sample some packets and delay them for 5ms in order to change the original packet order. Similarly for retransmitted packets, we randomly sampled packets to be duplicated.

We report the \emph{median} (and error bar) for ten independent runs.
In terms of accuracy, we consider the \emph{recall} and \emph{precision}, defined as $\frac{|ReturnedFlows \cap RelevantFlows| }{|RelevantFlows|}$ and $\frac{|EstimatedTrue|}{|True|}$. 

\subsection {Accuracy and Memory}
We first evaluate the accuracy of our Algorithms 1 to 4. 
In Algorithm~\ref{alg:rtt}, we detect the top 100 flows with the largest injected latency and report the recall rate. In Algorithm~\ref{alg:lost}, we track the top 100 flows with the largest injected random packet loss. Similarly, in Algorithms~\ref{alg:oopacket} and~\ref{alg:retransmissions}, we return the top 100 flows with a high number of out-of-order packets and retransmitted packets.

\para{Latency:} We track only TCP packets and replay the PCAP file with injected latency times for selected flows. As shown in Figure~\ref{fig:latency}, Algorithm 1 returns $>80\%$ of top 100 flows with 30ms average latency. When the latency is more significant, Algorithm 1 finds $\sim$98\% of top 100 high latency flows. If we get a slightly loose memory requirement, we can trade a small increase in memory for better accuracy. Figure~\ref{fig:latency_tradeoff} shows the recall rate of detecting top 100 flows with 50ms mean latency can reach $\sim$98\% when using 80KB memory.

\para{Packet loss:} As depicted in Figure~\ref{fig:packet_loss}, Algorithm 2 can detect the majority of the flows with the highest packet loss using 40KB memory. The recall rate increases significantly when the loss rate increases from low (1\%) to high (10\%).  In Figure~\ref{fig:packet_loss_tradeoff}, we show a similar memory-accuracy tradeoff as recall rate is improved by using more memory, and we can achieve $\sim$95\% recall rate when allocating 320KB memory.

\para{Out-of-order packets:} Using the injected real-world packet trace, Algorithm 3 can detect the flows with a high number of out-of-order packets. As shown in Figure~\ref{fig:packet_outoforder}, with 5 rows of 2000 counters our algorithm is able to detect the top 100 problematic flows among 3 million flows with $>95\%$ recall and $>92\%$ precision. When we increase the number of counters, there is a clear tradeoff between memory usage and accuracy, i.e., when using more than 80KB memory, $99\%$ of the flows with $4\%$ out-of-order rate have been returned. 

\para{Retransmission:} With a synthetic trace, we manually create duplicated packets with probabilities from 1\% to 10\%. Algorithm~\ref{alg:retransmissions} tracks elephant flows first and estimates the number of transmitted packets. The heavy-hitter algorithm detects the elephant flows with fairly perfect accuracy with even 40KB memory, as shown in Figures~\ref{fig:retransmission} and~\ref{fig:retransmission_tradeoff}. 

\begin{table}[t]
\centering
\footnotesize
\begin{tabular}{ lccccc }
\toprule
 \textbf{Switches} & Match Entries & Hash Bits & SRAMs & Action Slots\\
\midrule
Switch.p4 &  804 &  1678 &  293 & 503 \\
Our Prototype & 349 &  1051 & 31  & 98 \\
\bottomrule
\end{tabular}
\vspace{-4mm}\caption{Hardware resource usage on PISA switch.}
\label{tab:resource}
\end{table}

\subsection{Hardware Resources}\label{sec:eval:hardware}
Finally, we measure the resource usage of a hardware switch. The Protocol Independent Switch Architecture (PISA) we use allows developers to define their own packet formats and design the packet actions in a series of match-action tables. These tables are mapped into different stages in a sequential order, along with dedicated resources (e.g., match entries, hash bits, SRAMs, and action slots) for each stage. Our prototype leverages stateful memory to maintain the sketch data structure, and minimizes the resource usage. Table~\ref{tab:resource} shows small resource usage with Algorithms 1-4 combined, compared to a simple switch implementation provided by default (switch.p4).

\section{Related Work}\label{sec:related}
\para{End-host based monitoring tools:} With the full access to the end-host's network stack, existing work has tackled flow monitoring~\cite{tcpdump,tcpprobe}, event triggering~\cite{trumpet}, trace replay~\cite{deter,rack}, and performance monitoring and diagnosis ~\cite{lossradar,Confluo}. The advantage of using an end-host based approach is the accurate analysis of the network traffic. However, the deployment of such end-host based tools requires the control of the end-hosts, which largely limit the usage to private cloud environments. Recent work~\cite{SwitchPointer} also leverage switches to efficiently point to distributed end-host information for whole network visibility.

\para{Switch-based monitoring tools:} Since hardware switches have limited memory resources for monitoring tasks, memory-optimized sketching algorithms have been proposed to a variety of flow monitoring tasks, such as heavy hitters (frequent flows)~\cite{MG,SpaceSavings,revsketch,HashPipe}, detecting hierarchical heavy hitters~\cite{cormodeHHH,RHHH}, counting distinct flows~\cite{F0sketch,univmon}, estimating frequency moments~\cite{ams}, and change detection~\cite{k-ary,univmon}. These sketching algorithms offer worst-case guarantees to arbitrary network workloads and use sublinear memory in terms of the number of distinct network flows. 
On the other hand, virtual switches are emerging as an important measurement vantage point. Other than the traditional sampling-based approach~\cite{netflow,sflow}, recent work~\cite{hotnets_hashtable,nitro,SketchVisor,interval_query} has targeted on efficient data structures for flow monitoring in software switches. The main goal of these approaches is to achieve line-rate with accurate measurement results.
In contrast to past work on \emph{flow} monitoring, we propose sublinear data structures for \emph{performance} monitoring.

\section{Conclusion and Discussion}\label{sec:conclude}
In this paper, we propose  memory-efficient approaches for network performance monitoring. Our theoretical analysis and empirical evaluation demonstrate that our approaches achieve good accuracy with significant memory efficiency.  We conclude by highlighting a subset of new and exciting future work that this work opens up.

\para{Additional performance analytics:} 
Beyond tracking latency, packet loss, out-of-order, and retransmitted packets, we will consider more metrics, such as high delayed ACKs and low sending window. 
We have shown that there are memory lower bounds of $\Omega(N)$ for a workload with $N$ flows when deterministically obtaining some statistics in the model. 
Thus, to achieve lean memory efficiency, we need to identify if a particular performance monitoring task meets both flow-additive property and single flow sublinearity. If a performance function fails to meet either of these two properties, we need to further relax the problems with additional assumptions that are reasonable under practical networking scenarios.

\para{Hardware optimization:} When monitoring multiple statistics simultaneously on a hardware device, we would like to reduce the cost of maintaining multiple data structures. One simple way to reduce the hash bits (for other concurrent applications) is to store the hash value as metadata and reuse across sketches. We will further explore other probabilistic and succinct data structures specifically designed for performance monitoring in resource-constrained hardware.

\para{Possible universal data structure:} In this work, we define a computation model to describe the performance statistics collected from each flow and show the existence of lean algorithms. This unified model lights up a potential path to a ``universal sketch'' on all performance monitoring functions defined in the model. One attempt to build such a universal sketch can be maintaining a similar structure as UnivMon~\cite{univmon,liuhotnets15} or ElasticSketch~\cite{elastic} on each of the heavy flows.

\section{Acknowledgments}
We would like to thank the anonymous reviewers for their thorough comments and feedback that helped improve the paper. This work is supported in part by NSF grants CNS-1700521, CNS-1565343, CCF-1535948, CNS-1813487, NSF CAREER-1652257, Intel Labs University Research Office, ONR Award N00014-18-1-2364, the Lifelong Learning Machines program from DARPA/MTO, the Technion Hiroshi Fujiwara Cyber Security Research Center, and the Israel National Cyber Directorate.


\bibliographystyle{ieeetr}
\bibliography{alan}

\end{document}